\documentclass[UKenglish,thm-restate,authorcolumns]{lipics-v2019}
\usepackage{stackrel}
 \usepackage{wrapfig}
\usepackage{verbatim}
\usepackage{color}
\usepackage{xspace}
\usepackage{epsfig}
\usepackage{savesym}
\usepackage{tabularx}
\savesymbol{iint}
\savesymbol{iiint}
\usepackage{amsmath,amssymb}
\usepackage{tablefootnote}
\usepackage{wasysym}
\usepackage{pstricks,graphicx}
\usepackage{stmaryrd}
\usepackage{enumitem}
\usepackage{longtable}
\usepackage{pgf}
\usepackage{bussproofs}
\usepackage{soul}
\usepackage{float,subfloat}

\usepackage{tikz}
\usetikzlibrary{fit}
\usetikzlibrary{backgrounds}
\usetikzlibrary{patterns,calc}
\usetikzlibrary{decorations.pathreplacing}
\usetikzlibrary{positioning}
\usetikzlibrary{decorations.pathmorphing}
\usetikzlibrary{decorations.markings}
\tikzstyle{background}=[rectangle,fill=gray!10, inner sep=0.1cm, rounded corners=0mm]
\usepgflibrary{shapes}
\usetikzlibrary{snakes,automata,arrows}
\usetikzlibrary{shadows}
\tikzstyle{nloc}=[draw, text badly centered, rectangle, rounded corners, minimum size=2em,inner sep=0.5em]
\tikzstyle{background}=[rectangle,fill=gray!10, inner sep=0.1cm, rounded corners=0mm]
\tikzstyle{loc}=[draw,rectangle,minimum size=1.4em,inner sep=0em]
\tikzset{
    gluon/.style={decorate,draw=black,
        decoration={coil,amplitude=1pt, segment length=5pt}} 
}

\tikzset{
    gluonew/.style={decorate,draw=black,
        decoration={coil,amplitude=1pt, segment length=2pt}} 
}

\tikzset{
    gluon1/.style={decorate,draw=black,
        decoration={coil,amplitude=3pt, segment length=3pt}} 
}

\newcommand{\Red}[1]{{\textcolor{red}{#1}}}
\newcommand{\op}{\mathrm{op}}
\newcommand{\nop}{\mathrm{nop}}
\newcommand{\push}{\mathrm{push}}
\newcommand{\pop}{\mathrm{pop}}

\usepackage{algorithm}
\usepackage{algpseudocode}

\tikzset{
link/.style = {
     -stealth,
     shorten >=1pt
     },
array element/.style = {
    draw, fill=white,
    minimum width = 6mm,
    minimum height = 10mm
  },
  node of list/.style = { 
             draw, 
             fill=orange!20, 
             minimum height=6mm, 
             minimum width=6mm,
             node distance=10mm
   },
   ptr of list/.style = { 
             draw, 
             fill=orange!20, 
             minimum height=6mm, 
             minimum width=1mm,
             node distance=0mm
   },
}
\tikzset{
    ->,
    >=stealth',
    node distance=3cm,
    every state/.style={thick, fill=gray!10},
    initial text=$ $,
    }

\def\LinkedList#1{
  \foreach \element in \list {
     \node[node of list, right=15mm of aux, name=ele] {\element};
     \node[ptr of list, left=1mm of ele, name=mmp] {.};
     \node[ptr of list, left=5mm of ele, name=mmpp] {.};
     \node[ptr of list, left=9mm of ele, name=mm] {.};
     \draw[link] (aux) -- (mm);
     \coordinate (aux) at (ele.east);
  } 
}
\def\LinkedListTwo#1{
  \foreach \element in \list {
     \node[node of list, right= of aux1, name=ele] {\element};
     \draw[link] (aux1) -- (ele);
     \coordinate (aux1) at (ele.east);
  } 
}

\usepackage{lscape}

\usepackage{todonotes}

\newcommand{\cA}{\mathcal{A}}
\newcommand{\ZG}{\mathit{ZG}}

\newcommand{\cS}{\mathfrak{S}}

\newcommand{\TS}{\mathit{TS}}

\newcommand{\TLM}{\mathsf{TLM}}
\newcommand{\Todo}{\mathsf{ToDo}}

\newcommand{\create}{\mathsf{create}}
\newcommand{\add}{\mathsf{add}}
\newcommand{\addpush}{\mathsf{addPush}}
\newcommand{\addpop}{\mathsf{addPop}}
\newcommand{\isnewroot}{\mathsf{isNewRoot}}
\newcommand{\isnewnode}{\mathsf{isNewNode}}
\newcommand{\isnewpop}{\mathsf{isNewPop}}
\newcommand{\isnewpush}{\mathsf{isNewPush}}
\newcommand{\iterpop}{\mathsf{iterPop}}
\newcommand{\iterpush}{\mathsf{iterPush}}
\newcommand{\true}{\mathsf{true}}
\newcommand{\false}{\mathsf{false}}

\DeclareRobustCommand{\preceqv}{\mathrel{\rotatebox[origin=c]{-90}{$\preceq$}}}

\title{Fast zone-based algorithms for reachability in pushdown timed automata}

\author{S. Akshay}{Department of CSE, Indian Institute of Technology Bombay, Mumbai, India}{akshayss@cse.iitb.ac.in}{https://orcid.org/0000-0002-2471-5997}{}
\author{Paul Gastin}{Universit\'e Paris-Saclay, ENS Paris-Saclay, CNRS, LMF, 91190, Gif-sur-Yvette, France}{paul.gastin@lsv.fr}{https://orcid.org/0000-0002-1313-7722}{}
\author{Karthik R. Prakash}{Department of CSE, Indian Institute of Technology Bombay, Mumbai, India}{karthikrprakash@cse.iitb.ac.in}{https://orcid.org/0000-0003-4304-1382}{}
\authorrunning{S. Akshay, P. Gastin and K. R. Prakash}

\Copyright{S. Akshay, Paul Gastin and Karthik R. Prakash}

\ccsdesc[500]{Theory of computation~Quantitative automata}
\ccsdesc[500]{Theory of computation~Logic and verification}
\ccsdesc[300]{Theory of computation~Timed and hybrid models}

\keywords{Timed systems, Zone-abstractions, graphs}
\funding{Partly supported by ReLaX CNRS IRL 2000, DST/CEFIPRA/INRIA project EQuaVE and SERB Matrices grant MTR/2018/00074.}
\relatedversion{An extended abstract of this paper will appear in the proceedings of CAV'2021.}
\supplement{Our prototype implementation is available at \url{https://github.com/karthik-314/PDTA_reachability} and benchmarks for evaluation are at \url{https://doi.org/10.6084/m9.figshare.14501496.v2}.}
\hideLIPIcs  

\begin{document}    
\nolinenumbers
\maketitle

\begin{abstract}
  Given the versatility of timed automata a huge body of work has evolved that considers
  extensions of timed automata.  One extension that has received a lot of interest is  timed automata with a, possibly unbounded, stack, also called pushdown timed automata
  (PDTA).  While different algorithms have been given for reachability in different
  variants of this model, most of these results are purely theoretical and do not give
  rise to efficient implementations.  One main reason for this is that none of these
  algorithms (and the implementations that exist) use the so-called zone-based
  abstraction, but rely either on the region-abstraction or other approaches, which are significantly harder to implement.

  In this paper, we show that a naive extension, using simulations, of the zone based  reachability algorithm for the control state reachability problem of timed automata is  not sound in the presence of a stack.  To understand this better we give an inductive rule based view of the zone reachability algorithm for timed automata.  This alternate view allows us to analyze and adapt the rules to also work for pushdown timed automata. We obtain the first zone-based algorithm for PDTA which is terminating, sound and complete.  We implement our algorithm in the tool TChecker and perform experiments to show its efficacy, thus leading the way for more practical approaches to the verification of timed pushdown systems.
\keywords{Timed automata, Zone-based abstractions, Pushdown automata, Simulations, Reachability}
\end{abstract}
  
\section{Introduction}
\label{sec:intro}
Timed automata~\cite{alur1994theory} are a popular formalism for capturing real-time
systems, and of use for instance, in model checking of cyber-physical systems.  They
extend finite automata with real variables called clocks whose values increase over time;
transitions are guarded by constraints over these variables.  The main problem of interest
is the reachability problem, which asks whether a given state can be reached while
satisfying the constraints imposed by the guards.  This problem is known to be
PSPACE-complete (already shown in~\cite{alur1994theory}).  The PSPACE
algorithm, uses the so-called region-automaton construction, which essentially abstracts
the timed automaton into an exponentially larger finite automaton of regions (collections
of clock valuations), which is sound and complete for reachability.

Despite this complexity-theoretic hardness, the model of timed automata has proved to be
extremely influential and versatile, resulting in an enormous body of work on its theory,
variants and extensions over the past 25 years.  Almost since its inception, researchers
also began to develop tools to extend from theoretical algorithms to solve practical
problems.  Such tools range from the classical and richly featured tool
UPPAAL~\cite{bengtsson1995uppaal,larsen1997uppaal} to the more recent open-source tool
TChecker~\cite{tchecker}, which have been used on industry strength benchmarks and perform
rather well on many of them.  These tools use a different algorithm for reachability,
where reachable sets of valuations are represented as zones and explored in a graph.
While a naive exploration of zones does not terminate, the algorithms used identify
different strategies~\cite{BBLP04,HSW12,HKSW11}, e.g., subsumption or simulations,
extrapolations, for pruning the zone-based exploration graphs, while preserving soundness
and completeness of reachability.  While this does not change the worst case complexity,
in practice, the zone exploration results in much better {\em practical} performance as it
allows on-the-fly computation of reachable zones.  One could even argue that the wider
adoption of timed automata paradigm in the verification community has been a result of
scalable implementations and tools built on this zone-based approach.

In light of this, zone-based algorithms are often looked for to improve practical
performance of extensions of timed automata as well.  For instance, for timed automata
with diagonal constraints, classical zone-based approaches were shown to be
unsound~\cite{Bouyer04,BLR05}, but recently, an approach has been developed which adapts
the existing construction and obtains fast zone-based algorithms~\cite{GMS19}.  In the
present paper, we are interesting in adding a different feature to timed automata, namely
an unbounded lifo-stack.  This results in a powerful model of {\em pushdown timed automata
(PDTA for short)}, in which the source of ``infinity'' is both from real-time and the
unbounded stack.  Unsurprisingly, this model and its variants have been widely studied
over the last 20 years with several old and recent results on decidability of
reachability, related problems and their complexity,
including~\cite{AbdullaAS12,AGJK19,AGK16,AGKR20,AkshayGKS17,bouajjani1994automatic,ClementeL15,CL21,CLLM17,Dang03}.
A wide variety of techniques have been employed to solve these problems, from region-based
abstractions, to using atoms and systems of constraints, to encoding into different logics
etc.  However, except for~\cite{AkshayGKS17,AGKR20}, to the best of our knowledge, none of
the others carry an implementation.  In~\cite{AkshayGKS17}, the implementation uses a
tree-automaton implicitly based on regions and the focus in~\cite{AGKR20} is towards
multi-pushdown systems.  A common factor of all these works is that none of them consider
zone-based abstractions.

In this paper, we ask whether zone-based abstractions can be used to decide efficiently
reachability questions in PDTA. We focus on the problem of well-nested control-state
reachability of PDTA, i.e., given a PDTA, an initial and a target state, does there exist
a run of the PDTA that starts at the initial state with empty stack and reaches the target
state with an empty stack (in between, i.e., during the run, the stack can indeed be
non-empty).  As with timed automata, our goal here is towards its applicability to build
powerful tools which could lead to wider adoption of the PDTA model and showcase its
utility to model-checking timed recursive systems.  As the first step, we examine the
difficulties involved in mixing zones with stacks and point out that a naive adaptation of
the zone-based algorithm would not be sound.  Then we propose a new algorithm that
modifies the zone-based algorithm to work for pushdown timed automata.  This is done in
three steps.
\begin{itemize}[nosep]
  \item First we view the zone-graph exploration at the heart of the zone-based
  reachability algorithm for timed automata as a least fixed point computation of two
  inductive rules.  When applied till saturation, they compute a sound and complete finite
  abstraction of the set of all reachable zones.

  \item Next, this view allows us to generalize the approach in the presence of a stack by
  adding new inductive rules that correspond to push and pop transitions, and hence are
  specific to the stack operation.  There are two main technical difficulties in this.
  First, we need to ensure termination of the fixed point computation, using a strong
  enough pruning condition of the (a priori infinite) zone graph to ensure finiteness,
  while being sound and not adding spurious runs.  Second, we want to aggressively
  prune the graph as much as possible to obtain an efficient zone-exploration algorithm.
  We show how we can minimally change the condition of pruning in the zone exploration
  graph to achieve this delicate balance.  Indeed, in doing so we use a judicious
  combination of the subsumption (or simulation) relation and an equivalence relation for
  obtaining a fixed point computation for PDTA that is terminating, while being sound and
  complete.

  \item Finally, we build new data structures that allow us
  to write an efficient algorithm that implements this fixed point computation.  While
  getting a correct algorithm is relatively simple, to obtain an efficient one, we must
  again encounter and overcome several technical difficulties.
\end{itemize}
We implement our approach to build the first zone-based tool that efficiently solves
well-nested control state reachability for PDTA. Our tool is built on top of existing
infrastructure of TChecker~\cite{tchecker}, an open source tool and benefits from many
existing optimizations.  We perform experiments to show the practical performance of
multiple variants of our algorithm and show how our most optimized version is vastly
better in performance than other variants and of course the earlier region-based approach
on a suite of example benchmarks.

We note that our PDTA model differs slightly from the model considered
in~\cite{AbdullaAS12,AGK16}, as there is no age on stack and time spent on stack cannot be
compared with clocks.  Hence our model is closer to~\cite{bouajjani1994automatic,Dang03}.
However, in~\cite{ClementeL15}, it was shown that these two models are equivalent, more
specifically, the stack can be untimed without loss of expressivity (albeit with an
exponential blowup).  Thus our approach can be applied to the other model as well by just
untiming the stack.  There are other more powerful extensions~\cite{CLLM17,CL21} studied
especially in the context of binary reachability, where only theoretical results are
known.  We also remark that the idea of combining the subsumption relation between zones
with an equivalence relation also occurs while tackling liveness, or Buchi acceptance, in
timed automata.  This has been studied in depth~\cite{Tripakis09,Laarman13,HSTW20}, where
the naive zone-based algorithm does not work, forcing the authors to strengthen the
simulation relation in different ways.  Though these problems are quite different, there
are surprising similarities in the issues faced, as explained in
Section~\ref{sec:problem}.

The structure of the paper is as follows: we start with preliminaries and move on to the
difficulty in using zones and simulation relations in solving reachability in PDTA. Then,
we introduce in Section~\ref{sec:rewrite-rules} our inductive rules for timed automata and
PDTA and show their correctness.  In Section~\ref{sec:algo}, we present our algorithm and
helpful data-structural advancements.  We show the experimental performance in
Section~\ref{sec:experiments} and end with a brief conclusion.

\section{Preliminaries}
\label{sec:prelim}

\subsection{Timed automata}
{Timed automata} extend finite-state automata with a set $X$ of (non-negative) real-valued
variables called \emph{clocks}.  We let $\Phi(X)$ denote the set of constraints $\varphi$
that can be formed using the grammar: $\varphi::= x\sim c\mid x-y\sim c \mid
\varphi\wedge\varphi$, where $x,y\in X$, $c\in \mathbb{N}$, ${\sim}\in \{\leq,\geq,<,>\}$,
where constraints of the form $x\sim c,x-y\sim c$ are called atomic constraints. A clock valuation is a map $v\colon
X\rightarrow \mathbb{R}_{\geq0}$ and is said to satisfy $\varphi$, denoted $v\models
\varphi$, if $\varphi$ evaluates to true when each clock $x\in X$ is replaced with $v(x)$.
For $\delta\in \mathbb{R}^{\geq 0}$, we write $v+\delta$ to denote the valuation defined
as $(v+\delta)(x)=v(x)+\delta$ for all clocks $x$.  For $R\subseteq X$, we write $[R]v$ to
denote the valuation obtained by resetting clocks in $R$, i.e., $([R]v)(x)=0$ if $x\in R$,
and $([R]v)(x)=v(x)$ otherwise.  Finally, $v_0$ is the valuation that sets all
clocks to $0$.

A timed automaton $\cA$ is a tuple $(Q,X,q_0,\Delta,F)$, where $Q$ is a finite set of
states, $X$ is a finite set of clocks, $q_0\in Q$ is an initial state, $F\subseteq Q$ is
the set of final states and $\Delta\subseteq Q\times \Phi(X)\times 2^X\times Q$ is a set
of transitions.  A transition $t\in \Delta$ is of the form $(q,g,R,q')$, where $q,q'$ are
states, $g\in \Phi(X)$ is the guard of the transition and $R\subseteq X$ is the set of
clocks that are reset at the transition.
The semantics of a timed automaton $\cA$ is given as a transition system $\TS(\cA)$ over
configurations.  A configuration is a pair $(q,v)$ where $q\in Q$ is a state and $v$ is a
valuation, with the initial configuration being $(q_0,v_0)$.  The transitions are of two
types.  First, for a configuration $(q,v)$ and $\delta\in \mathbb{R}^{\geq0}$,
$(q,v)\xrightarrow{\delta} (q,v+\delta)$ is a delay transition.  Second, for
$t=(q,g,R,q')\in \Delta$, $(q,v)\xrightarrow{t} (q',v')$ is a discrete transition if
$v\models g$ and $v'=[R](v)$.  A run is an alternating sequence of delays and discrete
transitions starting from the initial configuration, and is said to be accepting if the
last state in the sequence is a final state. A timed automaton is said to be \emph{diagonal-free} if the atomic constraints used in guards do not involve diagonal constraints, i.e., constraints of the form $x-y\sim c$.

\subsection{Reachability, Zones and simulations}
\label{sec:prelim-zone-reachability}
The problem of control-state reachability asks whether a given timed automaton has an
accepting run.  This problem is known to be PSPACE-complete~\cite{alur1994theory}, originally shown
via the so-called region abstraction.  Note that, since $\TS(\cA)$ is infinite, some
abstraction is needed to get an algorithm.  In practice however, the abstraction used to
solve reachability, e.g., in tools such as UPPAAL~\cite{larsen1997uppaal} or TChecker~\cite{tchecker} is the
\emph{zone abstraction}.  A zone $Z$ is defined as a set of valuations defined by a
conjunction of atomic clock constraints.  Given a guard $g$ and reset $R$, we define the
following operations on zones: time elapse $\overrightarrow{Z}=\{v+\delta\mid v\in
Z,\delta\in \mathbb{R}^{\geq 0}\}$, guard intersection $g\cap Z=\{v\in Z\mid v\models g\}$
and reset $[R]Z=\{[R]v\mid v\in Z\}$.  The resulting sets are also zones.
With this, we can define the zone graph $\ZG(\cA)$ as a transition system
obtained
as follows: the nodes are (state, zone) pairs and $(q,Z)\xrightarrow{t}(q',Z')$,
if $t=(q,g,R,q')$ is a transition of $\cA$ and $Z'=\overrightarrow{[R](g\cap
Z)}$.  The initial node is $(q_0,Z_0=\overrightarrow{\{v_0\}})$ and a path in the
zone graph is said to be accepting if it ends at an accepting state.  The zone
graph is known to be sound and complete for reachability, but as the graph may
still be infinite, this does not give an algorithm for solving reachability yet.

To obtain an algorithm, one resorts to different techniques such as extrapolation or simulation.
Here we focus on \emph{simulation relations} which will lead to finite abstractions.
Given a timed automaton $\cA$, a binary relation $\preceq$ on configurations is called a
\emph{simulation} if whenever $(q,v)\preceq (q',v')$, we have $q=q'$ and
\begin{itemize}[nosep]
  \item for each delay $\delta\in \mathbb{R}^{\geq 0}$,
  $(q,v+\delta)\preceq(q,v'+\delta)$ and

  \item for each $t=(q,g,R,q_1)\in\Delta$, if $v\models g$ 
  then $v'\models g$ and $(q_1,[R]v) \preceq (q_1,[R]v')$.
\end{itemize}

We often simply write $v\preceq_q v'$ instead of $(q,v)\preceq(q,v')$.  We can now lift
this to sets $Z,Z'$ of valuations as $Z\preceq_q Z'$ if for all $v\in Z$ there exists
$v'\in Z'$ such that $v\preceq_q v'$.  We say that node $(q,Z)$ is subsumbed by node
$(q,Z')$ when $Z\preceq_q Z'$.  As a consequence we obtain the following lemma.

\begin{lemma}\label{lem:zone-simulation}
  If $(q,Z)\xrightarrow{t}(q_1,Z_1)$ in $\ZG(\cA)$ and $Z\preceq_q Z'$, then
  $(q,Z')\xrightarrow{t}(q_1,Z'_1)$ and $Z_1\preceq_{q_1}Z'_1$.
\end{lemma}

\begin{proof}
  Indeed, let $v_1\in Z_1=\overrightarrow{[R](g\cap Z)}$.  We find $v\in Z$ and
  $\delta\geq0$ such that $v\models g$ and $v_1=[R]v+\delta$.  Since $Z\preceq_q Z'$, we
  find $v'\in Z'$ with $v\preceq_q v'$.  We deduce that $v'\models g$ and
  $[R]v\preceq_{q_1}[R]v'$, which implies $v_1\preceq_{q_1}v'_1$ with
  $v'_1=[R]v'+\delta\in Z'_1=\overrightarrow{[R](g\cap Z')}$.
\end{proof}

A simulation $\preceq$ is said to be {\em finite} if for every sequence of nodes
$(q_1,Z_1)$, $(q_2,Z_2),\ldots$ there exist two nodes $(q_i,Z_i)$ and $(q_j,Z_j)$ with
$i<j$ such that $q_i=q_j$ and $Z_j\preceq_{q_i}Z_i$.  The importance of the finiteness is
that it allows us to stop exploration of zones along a branch of the zone graph: when a
node $(q_j,Z_j)$ is reached which is subsumed by an earlier node $(q_i,Z_i)$, we may cut
the exploration since all control states reachable from the latter are already reachable
from the former.  For a timed automaton $\cA$, we call this pruned graph as
$ZG_\preceq(A)$.  Thus, if the simulation relation $\preceq$ is finite, then
$ZG_\preceq(\cA)$ is finite, sound and complete for control state reachability.  
We formalize this algorithm in Section~\ref{sec:rewrite-rules}, using inductive rules.

Various finite simulations have been shown to exist in the literature, including the
famous LU-abstractions~\cite{BBLP04} for diagonal-free timed automata, and more recent $\mathcal{G}$-abstractions based on sets of guards \cite{GMS19} for timed automata with diagonal constraints.  Hence this theory indeed has resulted in better implementations and is used in standard tools in this domain.

We will see that using simulation in the context of pushdown timed automata is not always 
sound, in some cases we need a stronger condition to stop the exploration. Towards this, 
we consider the equivalence relation on nodes induced by the simulation relation:
$Z\sim_q Z'$ if $Z\preceq_q Z'$ and $Z'\preceq_q Z$. We say that the 
simulation $\preceq$ is strongly finite if the induced equivalence relation $\sim$ has 
finite index. Notice that strongly finite implies finite but the converse does not 
necessarily hold. Fortunately, the usual simulations for timed automata, in particular 
the LU-simulation and the $\mathcal{G}$-simulation, are strongly finite.

\subsection{Pushdown timed automata (PDTA)}
A Pushdown Timed Automaton $\cA$ is a tuple $(Q,X,q_0,\Gamma,\Delta,F)$, where $Q$ is a finite set of states, $X$ is
a finite set of clocks, $q_0\in Q$ is an initial state, $\Gamma$ is the stack alphabet,
$F\subseteq Q$ is the set of final states and $\Delta$ is a set of transitions.  A
transition $t\in \Delta$ is of the form $(q,g,\op,R,q')$, where $q,q'$ are states, $g\in
\Phi(X)$ is the guard of the transition and $R\subseteq X$ is the set of clocks that are
reset at the transition, $\op$ is one of three stack operations: $\nop$ or $\push_a$ or
$\pop_a$ for some $a\in \Gamma$.

The semantics of a PDTA $\cA$ is given as a transition system $\TS(\cA)$ over
configurations.  A configuration here is a tuple $(q,v,\chi)$ where $q\in Q$ is a state,
$v$ is a valuation, $\chi\in \Gamma^*$ is the stack content, with the initial
configuration being $(q_0,v_0,\varepsilon)$.  The transitions are of two types.  First,
for a configuration $(q,v,\chi)$ and $\delta\in \mathbb{R}^{\geq0}$,
$(q,v,\chi)\xrightarrow{\delta} (q,v+\delta,\chi)$ is a delay transition.  Second, for
$t=(q,g,\op,R,q')\in \Delta$, $(q,v,\chi)\xrightarrow{t} (q',v',\chi')$ is a discrete
transition if $v\models g$, $v'=[R](v)$ and
\begin{itemize}[nosep]
\item if $\op=\nop$, then $\chi'=\chi$,
\item if $\op=\push_a$ then $\chi'=\chi \cdot a$,
\item if $\op=\pop_a$, then $\chi=\chi'\cdot a$.
\end{itemize}
A run is an alternating sequence of delays and discrete actions starting from the initial
configuration. It is accepting if the last state in the sequence is final.

Our main focus is the {\em well-nested control state reachability} problem for PDTA, which
asks whether a configuration $(q,v,\varepsilon)$ with $q\in F$ is reachable, where the
stack is empty.  Later, in Section~\ref{sec:conclusion}, we remark how our solution can be
extended to solve general control state reachability, i.e., asking whether a configuration
$(q,v,\chi)$ with $q\in F$ is reachable, possibly with a nonempty stack $\chi$.

\section{Zones in PDTA and the problem with simulations}
\label{sec:problem}
As mentioned earlier, zones are collections of clock valuations defined by conjunctions of
timing constraints, and exploring zones reached by a timed automaton gives a sound and
complete abstraction for state reachability.  To make sure that the exploration is finite
we need to prune the graph and one way this is done by simulation, i.e., not exploring
paths from some nodes if they are ``subsumed'' by earlier nodes visited in the graph.
Consider Figure~\ref{fig:tpda}, in which we ignore the $\push_a$ and $\pop_a$ or we can
think of them as internal actions.  Then the usual zone-graph construction with
simulation would give the graph depicted in Figure~\ref{fig:zones}.  In this section,
just for illustration we instantiate the simulation relation to be the well-known
LU-simulation (we do not give the definition here as it is not relevant to what comes
later, instead we refer to earlier work~\cite{BBLP04}).  Using this, we obtain that the
rightmost node is subsumed by the previous one, 
and hence the dotted simulation edge.  If we did not do this we immediately observe that
we get an infinite graph with increasing sets of zones.

Now, our first question is whether this zone exploration with simulation can be lifted to
PDTA. In this example, if we were to add back the push/pop edges, we get
exactly the same Zone graph with annotations, and further, the final state is indeed
reachable.  Hence, for this particular example we do obtain a finite, sound and complete
graph exploration.  However, in general it turns out that the procedure is not sound.

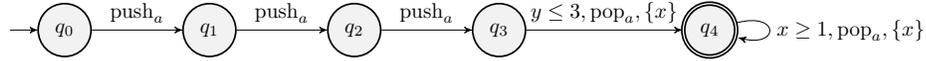
\begin{figure}[tb]
  \centering
\scalebox{0.8}{
  \begin{tikzpicture}[node distance = 3.5cm]
    \node[state, initial] (q0) {$q_0$};
    \node[state, right=1.5cm of q0] (q1) {$q_1$};
    \node[state, right=1.5cm of q1] (q2) {$q_2$};
    \node[state, right=1.5cm of q2] (q3) {$q_3$};
    \node[state, accepting, right of=q3] (q4) {$q_4$};

    \draw (q0) edge[] node[above]{$\push_a$} (q1)
    (q1) edge[] node[above]{$\push_a$} (q2)
    (q2) edge[] node[above]{$\push_a$} (q3)
    (q3) edge[] node[above]{$y\leq 3,\pop_a,\{x\}$} (q4)
    (q4) edge[loop right] node[right]{$x\geq 1 ,\pop_a,\{x\}$} (q4);
  \end{tikzpicture}
  }
\caption{A simple PDTA with 2 clocks $\{x,y\}$. Note that if we ignore the push/pop actions we get a TA, say $A$.}
\label{fig:tpda}
\end{figure}
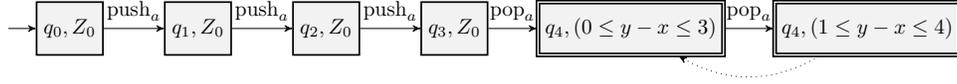
\begin{figure}[tb]
\centering
\hspace{-0.5cm}
\scalebox{0.8}{
\begin{tikzpicture}[node distance = 4cm]
    \node[state, rectangle, initial] (q0) {$q_0,Z_0$};  
    \node[state, rectangle, right=1cm of q0] (q1) {$q_1,Z_0$};
    \node[state, rectangle, right=1cm of q1] (q2) {$q_2,Z_0$};
    \node[state, rectangle, right=1cm of q2] (q3) {$q_3,Z_0$};
    \node[state, rectangle, accepting, right=0.8cm of q3] (q4) {$q_4,(0\leq y-x\leq 3)$};
    \node[state, rectangle, accepting, right=0.8cm of q4] (q5) {$q_4,(1\leq y-x\leq 4)$};

    \draw (q0) edge[] node[above]{$\push_a$} (q1)
    (q1) edge[] node[above]{$\push_a$} (q2)
    (q2) edge[] node[above]{$\push_a$} (q3)
    (q3) edge[] node[above]{$\pop_a$} (q4)
    (q4) edge[] node[above]{$\pop_a$} (q5)
    ;
    \draw[dotted] (q5) edge[bend left] node[left]{} (q4);
\end{tikzpicture}
}
\caption{Zone graph with simulation edges for finiteness.  Again ignoring push/pop actions
gives us a zone graph for the TA. $Z_0$ is the initial zone.}
\label{fig:zones}
\end{figure}

Consider the example in Figure~\ref{fig:tpda2}.  In this example, again considering it as
a TA (ignoring the push/pops), we would get the zone graph below, which would be finite,
sound and complete for reachability in that TA. But if we consider it as a PDTA, now doing
the same does not preserve soundness.  In other words, in the PDTA, $q_3$ is no longer
reachable.  However, in the zone graph we would conclude that it is reachable due to the
simulation edge.  If, to fix this, we remove the dotted simulation edge, then we will
lose finiteness.

\begin{figure}[tb]
\scalebox{0.8}{
  \begin{tikzpicture}[node distance = 3cm]
    \node[state, initial] (q0) {$q_0$};
    \node[state, right=1cm of q0] (q2) {$q_2$};
    \node[state, accepting, right=1cm of q2] (q3) {$q_3$};
    \node[state, below=2cm of q0] (q1) {$q_1$};

    \draw (q0) edge[bend left] node[right]{$x\geq1,\{x\}$} (q1)
    (q1) edge[bend left] node[left]{$y\leq1,\push_a$} (q0)
    (q0) edge[] node[above]{$\pop_a$} (q2)
    (q2) edge[] node[above]{$\pop_a$} (q3);
  \end{tikzpicture}
  }
\scalebox{0.8}{
\begin{tikzpicture}[node distance = 3cm]
  \node[state, rectangle, initial](q0) {$q_0,Z_0$};
    \node[state, rectangle, right=1cm of q0] (q2) {$q_2,Z_0$};
    \node[state, rectangle, below=2cm of q0] (q1) {$q_1,y-x\geq1$};
    \node[state, rectangle, accepting, right=1cm of q2] (q3) {$q_3,Z_0$};
    \node[state, rectangle, right=2cm of q1] (q4) {$q_0,y-x=1$};

    \draw (q0) edge[] node[right]{} (q1)
    (q1) edge[] node[above]{$\push_a$} (q4)
    (q0) edge[] node[above]{$\pop_a$} (q2)
    (q2) edge[] node[above]{$\pop_a$} (q3);
    \draw[dotted] (q4) edge[] node[left]{} (q0);
\end{tikzpicture}
}
\caption{A PDTA and its zone graph with simulation.  With the simulation (dotted)
edges, $q_3$ is reachable in the zone graph, but its not reachable in the PDTA.}
\label{fig:tpda2}
\end{figure}
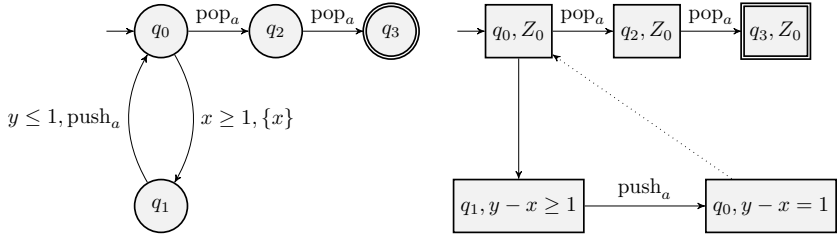

Thus, it seems that we have a difficult situation where zones with the simulation
relation, needed for termination, do not preserve soundness.  This situation resembles the
situation studied in~\cite{Tripakis09,Laarman13,HSTW20}, where the authors study liveness
or Buchi-acceptance conditions in timed automata.  Again in that situation, the naive
algorithm with zone simulation does not work and the authors are forced to strengthen the
simulation relation in different ways.

Surprisingly, it turns out, that even in our very different problem setting of
reachability in PDTA, a similar solution works.  That is, we replace simulation by
equivalence (defined in the previous section) as the pruning criterion.  However, there
are two issues (i) it is not easy to prove its correctness and (ii) this is far from
efficient as shown in the experimental section.  Our goal to use zones in the first place
was efficiency and hence we would like to prune the zone graph as much as possible, i.e.,
we would like to use simulation edges as much as possible.  In the next two sections, we
describe our fix.  We first show a different view of the exploration algorithm as a fixed
point rule based approach.  This allows us to then describe our fix in the same language,
which is much easier to understand conceptually.  Also as a corollary we will be able to
show that using equivalence everywhere also gives a correct algorithm.  After proving the
correctness of our rule-based algorithm, we then tackle the challenges in implementing it.

\section{Viewing reachability algorithms using rewrite rules}\label{sec:rewrite-rules}

In this section, our goal is to compute a set $S$ of nodes of the zone graph of a PDTA, as
a least fixed point of a small set of inductive rules, such that a control state $q$
occurs in $S$, i.e., $(q,Z)\in S$ for some $Z$ iff $q$ is reachable in the PDTA from its
initial state.  To understand the rules and their correctness it is easier to first
visualize this on plain timed automata without any push-pop edges.

\subsection{Rewrite rules for Timed automata.}
Given a TA $\cA=(Q,X,q_0,\Delta,F)$, the set $S$ containing all reachable
nodes of the zone graph, can be obtained as the least fixed point of the following
inductive rules, with a natural deduction style of presentation.

\begin{prooftree}
  \AxiomC{}
  \RightLabel{\scriptsize{start}}
  \UnaryInfC{$S:=\{(q_0,Z_0)\}$}
\end{prooftree}

\begin{prooftree}
  \AxiomC{$(q,Z)\in S$}
  \AxiomC{$q\xrightarrow{g,R}q'$}
  \AxiomC{$Z'=\overrightarrow{R(g\cap Z)}\neq\emptyset$}
  \RightLabel{\scriptsize{Trans}}
  \TrinaryInfC{$S:=S\cup \{(q',Z')\}$}
\end{prooftree}

Let $S^*$ denote the set at the fixed point by starting with the start rule and repeatedly
applying the trans rule.  It is easy to see that this computes the set of all reachable
nodes of the zone graph: the start rule starts with the initial node and each application
of trans rule takes a reachable node and applies a transition of the automaton and
includes the resulting node reached.  However, this set $S^*$ is a priori infinite since
number of zones is infinite.

To make it finite we add a condition under which we will apply the transition rule based on a finite simulation relation (let us denote it $\preceq$) for $\cA$.
\begin{prooftree}
  \AxiomC{$(q,Z)\in S$}
  \AxiomC{$q\xrightarrow{g,R}q'$}
  \AxiomC{$Z'=\overrightarrow{R(g\cap Z)}\neq\emptyset$}
  \RightLabel{\scriptsize{Trans-$\preceq$}}
  \TrinaryInfC{$S:=S\cup \{(q',Z')\}$,
  \Red{unless $\exists (q',Z'')\in S$, $Z'\preceq_{q'}Z''$}}
\end{prooftree}

Thus to obtain an algorithm, we would explore all nodes in the Zone graph using a search
algorithm (say DFS/BFS) and we would add a node only if it is not subsumed by an already
visited node, according to the simulation relation.  We explained in
Section~\ref{sec:prelim-zone-reachability} that doing this preserves soundness and
completeness and gives a finite exploration.
        
\begin{lemma}\label{lem:tarules}
  Let $S_\preceq^*$ denote any set obtained from the start rule and by repeatedly applying
  Trans-$\preceq$ till a fixed point is reached.  Note that depending on the order of
  applications we may have different sets.  Then we have:
  \begin{enumerate}
    \item (finiteness) $S_\preceq^*$ is finite.

    \item (soundness and completeness) For all $q\in Q$, a configuration $(q,v)$ is
    reachable from $(q_0,v_0)$ in the TA $\cA$ iff $(q,Z)\in S_\preceq^*$ for some
    zone $Z$.
  \end{enumerate}
\end{lemma}
We do not give the proof here as (i) it is only a reformulation of known
results and (ii) it will be subsumed by the much stronger theorem we prove next.

\subsection{Rewrite rules for PDTA.}\label{sec:rules-PDTA}
Let $\cA=(Q,X,q_0,\Gamma,\Delta,F)$ be a PDTA, we will need not just a
set but a tuple of sets.  More precisely, we maintain a set of nodes $\cS$ called
\emph{root} nodes.  For each root node $(q,Z)\in \cS$, we also maintain a set of nodes,
denoted $S_{(q,Z)}$.  The intuition is that root nodes are those that can be reached after
pushing a symbol to the stack, whereas $S_{(q,Z)}$ will be the set of nodes that can be
reached from $(q,Z)$ with a well-nested run, i.e., starting with an empty stack and ending
in an empty stack.  This is to avoid storing the stack contents in our algorithm,
which would be another source of infinity.  Again, we use simulations to make the
computation finite.  So we fix a strongly finite simulation relation $\preceq$ for $\cA$.

\begin{table}[tbh]
  \centering
  \begin{prooftree}
    \AxiomC{}
    \RightLabel{\scriptsize{Start}}
    \UnaryInfC{$\cS:=\{(q_0,Z_0)\}$, $S_{(q_0,Z_0)}:=\{(q_0,Z_0)\}$}
  \end{prooftree}  
  
  \begin{prooftree}
    \AxiomC{$(q,Z)\in \cS$}
    \AxiomC{$(q',Z')\in S_{(q,Z)}$}
    \AxiomC{$q'\xrightarrow{g,\nop,R}q''$}
    \AxiomC{$Z''=\overrightarrow{R(g\cap Z')}\neq\emptyset$}
    \RightLabel{\scriptsize{Internal}}
    \QuaternaryInfC{$S_{(q,Z)}:=S_{(q,Z)}\cup \{(q'',Z'')\}$, 
    \Red{unless $\exists (q'',Z''')\in S_{(q,Z)}$, $Z''\preceq_{q''}Z'''$}}
  \end{prooftree}
  
  \begin{prooftree}
    \AxiomC{$(q,Z)\in \cS$}
    \AxiomC{$(q',Z')\in S_{(q,Z)}$}
    \AxiomC{$q'\xrightarrow{g,\push_a,R}q''$}
    \AxiomC{$Z''=\overrightarrow{R(g\cap Z')}\neq\emptyset$}
    \RightLabel{\scriptsize{Push}}
    \QuaternaryInfC{$\cS:=\cS\cup \{(q'',Z'')\}$, $S_{(q'',Z'')}=\{(q'',Z'')\}$, 
    \Red{unless $\exists (q'',Z''')\in \cS$, $Z''\sim_{q''}Z'''$}}
  \end{prooftree}
  
  \begin{prooftree}
  \def\defaultHypSeparation{\hskip 0.2cm}
    \alwaysNoLine
    \AxiomC{$(q,Z)\in \cS$}
    \UnaryInfC{$(q'',Z_1)\in \cS$}
    \AxiomC{$(q',Z')\in S_{(q,Z)}$}
    \UnaryInfC{$(q'_1,Z'_1)\in S_{(q'',Z_1)}$}
    \AxiomC{$q'\xrightarrow{g,\push_a,R}q''$}
    \UnaryInfC{$q'_1\xrightarrow{g_1,\pop_a,R_1}q_2$}
    \AxiomC{$Z''=\overrightarrow{R(g\cap Z')}\sim_{q''}Z_1$}
    \UnaryInfC{$Z_2=\overrightarrow{R_1(g_1\cap Z'_1)}\neq\emptyset$}
    \alwaysSingleLine
    \RightLabel{\scriptsize{Pop}}
    \QuaternaryInfC{$S_{(q,Z)}:=S_{(q,Z)}\cup \{(q_2,Z_2)\}$,
    \Red{unless $\exists (q_2,Z'_2)\in S_{(q,Z)}$, $Z_2\preceq_{q_2}Z'_2$}}
  \end{prooftree}

  \caption{Inductive rules for well-nested control state reachability of PDTA}
  \label{tbl:rules-tpda}
\end{table}

Our inductive rules for well-nested control state reachability of pushdown timed automata are
given in Table~\ref{tbl:rules-tpda}.  Note that the internal rule is just the same as for
timed automata above.  The start rule not only starts the set of nodes computation but
also the set of roots computation as described above.  So the only interesting rules are
the Push and Pop rules.  The push rule says that when a push is encountered, then we must
start exploring from a new root (i.e., context).  So the only complicated rule is the Pop
rule.  Here the intuition is that if we see a push at a node and from a root equivalent to
the root created from it, (i.e., its context) we see a matching pop reaching a new node,
then this push-pop context is complete, and we can add this new node to the set of
reachable nodes.  This is precisely the point where we {\em need} equivalence rather than
simulation and this will be made clear in the proof of the theorem below.

\begin{theorem}\label{thm:tpdarules}
  Let $\cS^*$ and $(S_{(q,Z)})_{(q,Z)\in\cS^{*}}$ denote any tuple of sets obtained from
  the start rule and by repeatedly applying the rules in Table~\ref{tbl:rules-tpda} till a
  fixed point is reached\footnote{As before, there could be several such sets depending on
  the order in which the rules are applied.}.  Note that we always have
  $(q_0,Z_0)\in\cS^*$.  The following statements hold:
  \begin{enumerate}[nosep]
    \item (finiteness) $\cS^*$ is finite and for each
    $(q,Z)\in \cS^*$, $S_{(q,Z)}$ is finite.
    
    \item (completeness) For each $(q,Z)\in \cS^*$, if there exists a run
    $(q,v,\varepsilon)\xrightarrow{*} (q',v',\varepsilon)$ of $\cA$ with $\{v\}\preceq_q
    Z$, then there exists $(q',Z')\in S_{(q,Z)}$ such that $\{v'\}\preceq_{q'} Z'$.
    
    \item (soundness) For each $(q,Z)\in \cS^*$, $(q',Z')\in S_{(q,Z)}$ and $v'\in Z'$,
    there exists a run in PDTA from $(q,v,\varepsilon)$ to $(q',v'',\varepsilon)$ with
    $v\in Z$ and $v'\preceq_{q'}v''$.

  \end{enumerate}
\end{theorem}

\begin{proof}
    \textbf{1.} Note that only the Push rule creates new root nodes and the red condition
    states that a new root node is added only if there isn't already an equivalent node in
    $\cS^*$.  Since the simulation relation is strongly finite, the set of roots $\cS^{*}$
    must be finite.  Also, before adding a node to some $S_{(q,Z)}$ with the internal rule
    or the pop rule, we check that the node is not subsumed by an existing one.  Since the
    simulation relation is finite, this ensures that each set $S_{(q,Z)}$ is finite.
        
    \smallskip\noindent\textbf{2.} 
    Let $(q,Z)\in \cS^*$ and assume that $(q',v',\varepsilon)$ is reachable from some
    $(q,v,\varepsilon)$ with $v\preceq_q Z$, i.e., there exists a run
    $(q,v,\varepsilon)=(q_1,v_1,\chi_1)\rightarrow\cdots
    \rightarrow(q_n,v_n,\chi_n)=(q',v',\varepsilon)$.  We will then show that
    $v_n\preceq_{q_n}Z'$ for some $(q_n,Z')\in S_{(q,Z)}$.  The proof is by induction on
    $n$.
      Base case: For $n=1$ we have $q'=q$ and $v'=v$.  The result is obtained by
      taking $Z'=Z$.  Notice that $(q,Z)\in S_{(q,Z)}$ follows immediately from the start
      rule if $q=q_0$, $Z=Z_0$ or from the push-create rule.
      
      Let us then assume that the statement holds for runs of length at most $n-1$.
      Consider any run of the form $(q,v,\varepsilon)=(q_1,v_1,\chi_1)\rightarrow\cdots
      \rightarrow(q_n,v_n,\chi_n=\varepsilon)$ with $v\preceq_q Z$.  Notice that its
      last transition $(q_{n-1},v_{n-1},\chi_{n-1})\rightarrow (q_n,v_n,\chi_n)$ cannot be
      a push transition (in the PDTA) since $\chi_n=\varepsilon$.  Hence, we have three
      subcases, depending on the last transition.
      \begin{itemize}[nosep]
        \item Time elapse. $\chi_{n-1}=\chi_n=\varepsilon$,
        $q_{n-1}=q_n=q'$, $v_n=v_{n-1}+\delta$ for some $\delta\in \mathbb{R}^{\geq 0}$.
        Applying induction hypothesis, we have $v_{n-1}\preceq_{q'}Z'$ for some
        $(q',Z')\in S_{(q,Z)}$.  Since zones are closed under time elapse, we get
        $Z'=\overrightarrow{Z'}$ and by definition of the simulation relation
        $v_n=v_{n-1}+\delta\preceq_{q'}\overrightarrow{Z'}=Z'$.  This completes the case.
        
        \item Discrete internal transition.  In this case
        $\chi_{n-1}=\chi_n=\varepsilon$, $t=q_{n-1}\xrightarrow{g,\nop,R}q_n$,
        $v_{n-1}\models g$ and $v_n=[R]v_{n-1}$.  Then applying induction hypothesis,
        there exists $(q_{n-1},Z')\in S_{(q,Z)}$ such that $v_{n-1}\preceq_{q_{n-1}}Z'$.
        Now let $Z''=\overrightarrow{R(g\cap Z')}$.  From the definition of the simulation
        relation we get $v_n\preceq_{q_n}Z''$.  Then, applying the Internal rule, 
        there exists $(q_n,Z''')\in S_{(q,Z)}$ such that $Z''\preceq_{q_n}Z'''$, with
        possibly $Z'''=Z''$.  Hence, $v_n\preceq_{q_n}Z''\preceq_{q_n}Z'''$, which
        completes the case.

        \begin{figure}[tbp]
          \scalebox{0.78}
          {
          $
          \begin{array}{ccccccccc}
            (q,v,\varepsilon) & \xrightarrow{*} & (q_i,v_i,\varepsilon) & \xrightarrow{t} & 
            (q_{i+1},v_{i+1},a) & \xrightarrow{*} & (q_{n-1},v_{n-1},a) & \xrightarrow{t_1} & (q_n,v_n,\epsilon)  
            \\[1ex]
            v\preceq_{q}Z &  & v_i\preceq_{q_i}Z_i &  & v_{i+1}\preceq_{q_{i+1}}Z_{i+1} &  & 
            v_{n-1}\preceq_{q_{n-1}}Z_{n-1} &  & v_n\preceq_{q_n}Z_n 
            \\[1ex]
            (q,Z)\in \cS &  & (q_i,Z_i)\in S_{(q,Z)} &  & Z_{i+1}=\overrightarrow{R(g\cap Z_i)} &  & 
            (q_{n-1},Z_{n-1})\in S_{(q_{i+1},Z'_{i+1})} &  & Z_n=\overrightarrow{R_1(g_1\cap Z_{n-1})}
            \\[1ex]
            &  &  &  & Z_{i+1}\sim_{q_{i+1}}Z'_{i+1} &  &  &  & Z_n\preceq_{q_n}Z'_n
            \\[1ex]
            &  & &  & (q_{i+1},Z'_{i+1})\in \cS &  &  &  & (q_n,Z'_n)\in S_{(q,Z)} 
          \end{array}
          $
          }
          \caption{Construction for the completeness-push-pop last sub-case.}
          \label{fig:complete-pop}
        \end{figure}
        
        \item Pop transition.  Then there exists $1\leq i<n-1$ such that the
        run has the form: $(q_1,v_1,\varepsilon)\rightarrow\ldots \rightarrow
        (q_i,v_i,\chi_i=\varepsilon)\xrightarrow{\push_a}
        (q_{i+1},v_{i+1},\chi_{i+1}=a)\rightarrow \ldots \rightarrow
        (q_{n-1},v_{n-1},\chi_{n-1}=a)\xrightarrow{\pop_a} (q_n,v_n,\chi_n=\varepsilon)$,
        where the push and pop are matching transitions, i.e., $|\chi_j|\geq1$ for all
        $i<j<n-1$ (see Figure~\ref{fig:complete-pop}).  Then by induction hypothesis at
        $i$, we have
        \begin{align} 
          v_i\preceq_{q_i}Z_i \text{ for some } (q_i,Z_i)\in S_{(q,Z)} \,.
        \end{align}
        From the push transition we have
        \begin{align}
          \exists t=q_i\xrightarrow{g, \push_a, R} q_{i+1} \in\Delta 
          \text{ with } v_i\models g \text{ and } v_{i+1}=[R]v_i \,.
        \end{align}
        Let $Z_{i+1}=\overrightarrow{R(g\cap Z_i)}$.  By definition of the simulation
        relation, we deduce from $v_i\preceq_{q_i}Z_i$ that
        $v_{i+1}\preceq_{q_{i+1}}Z_{i+1}$.  We can apply the Push rule to obtain
        \begin{align}
          (q_{i+1},Z'_{i+1})\in\cS^* \text{ for some } 
          Z'_{i+1}\sim_{q_{i+1}}Z_{i+1}
        \end{align}
        possibly with $Z'_{i+1}=Z_{i+1}$ as a special case.
        
        Further the segment of run $(q_{i+1},v_{i+1},a)\rightarrow \ldots
        (q_{n-1},v_{n-1},a)$ in the PDTA never pops the symbol $a$ (by choice, since
        otherwise the push and pop would not be matching).  Hence we will also have the
        same sequence of transitions forming a run
        $(q_{i+1},v_{i+1},\varepsilon)\rightarrow \ldots (q_{n-1},v_{n-1},\varepsilon)$.
        Using $v_{i+1}\preceq_{q_{i+1}}Z_{i+1}\sim_{q_{i+1}}Z'_{i+1}$, we deduce that
        $v_{i+1}\preceq_{q_{i+1}}Z'_{i+1}$.
        By induction hypothesis,
        \begin{align}
          v_{n-1}\preceq_{q_{n-1}}Z_{n-1} \text{ for some }
          (q_{n-1},Z_{n-1})\in S_{(q_{i+1},Z'_{i+1})}  \,.
        \end{align}
        Finally, we have the pop transition
        \begin{align}
          t_1=q_{n-1}\xrightarrow{g_1,\pop_a,R_1} q_n \in\Delta 
          \text{ with } v_{n-1}\models g_1 \text{ and } v_n=[R_1]v_{n-1} \,.
        \end{align}
        We let $Z_n=\overrightarrow{R_1(g_1\cap Z_{n-1})}$.  From
        $v_{n-1}\preceq_{q_{n-1}}Z_{n-1}$ and the definition of the simulation relation we
        obtain $v_n\preceq_{q_n}Z_n$.  Then, combining all the above equations (1--5), and
        applying the Pop-rule we obtain some $(q_n,Z'_n)\in S_{(q,Z)}$ with
        $Z_n\preceq_{q_n}Z'_n$ (possibly $Z'_n=Z_n$).  Finally we get
        $v_n\preceq_{q_n}Z_n\preceq_{q_n}Z'_n$.  This completes the proof.
      \end{itemize}
    \smallskip\noindent\textbf{3.}
    We will show that the following property is invariant by rule applications:
    \begin{align}
      \forall (q,Z)\in \cS,~& \forall (q',Z')\in S_{(q,Z)}, \forall v'\in Z', 
      \text{ there is a run } 
      \notag\\
      &(q,v,\varepsilon)\xrightarrow{*}(q',v'',\varepsilon)
      \text{ with } v\in Z \text{ and } v'\preceq_{q'}v''
      \tag{Inv}
      \label{eq:invariant}
    \end{align}
    The invariant holds initially, i.e., after application of the start rule.
    Indeed, in this case we have $\cS=\{(q_0,Z_0)\}$ and $S_{(q_0,Z_0)}=\{(q_0,Z_0)\}$. 
    Hence $(q',Z')=(q,Z)=(q_0,Z_0)$ and for all $v\in Z_0$ we can choose 
    the empty run $(q_0,v,\varepsilon)\xrightarrow{0}(q_0,v,\varepsilon)$.
    
  We show below that \eqref{eq:invariant} is preserved by application of an
  internal/push/pop rule.  Therefore, the invariant still holds when reaching
  the fixed point, which proves the soundness.  
  Let us write $\cS^{-}$ and $S^{-}_{(q,Z)}$ for the sets before 
  the application of the rule and $\cS$ and $S_{(q,Z)}$ for the sets after 
  the application of the rule.
  \par
  \smallskip\noindent\textbf{Internal rule.}
    Let $(q,Z)\in\cS=\cS^{-}$, $(q',Z')\in S_{(q,Z)}$ and $v'\in Z'$. If 
    $(q',Z')\in S^{-}_{(q,Z)}$ then we get the result since 
    \eqref{eq:invariant} holds before applying the internal rule. Otherwise, 
    there is some $(q_1,Z_1)\in S^{-}_{(q,Z)}$ and a transition
    $t=q_1\xrightarrow{g,\nop,R}q'$ with $Z'=\overrightarrow{R(g\cap Z_1)}$.  
    
    By definition, there exists $v_1\in Z_1$ such that $v_{1}\models g$ and
    $v'=[R]v_{1}+\delta$ for some $\delta\geq0$.  Hence we have a run
    $(q_{1},v_{1},\varepsilon)\xrightarrow{t}\xrightarrow{\delta}(q',v',\varepsilon)$.
    Since the invariant holds before the internal rule, there is a run
    $(q,v,\varepsilon)\xrightarrow{*}(q_{1},v'_{1},\varepsilon)$ with $v\in Z$ and
    $v_{1}\preceq_{q_1}v'_{1}$.  Now since $\preceq$ is a simulation we obtain that
    $(q_{1},v'_{1},\varepsilon)\xrightarrow{t}\xrightarrow{\delta}(q',v'',\varepsilon)$
    with $v'\preceq_{q'}v''$ and we are done.

  \smallskip\noindent\textbf{Push rule.}
    Let $(q,Z)\in\cS$, $(q',Z')\in S_{(q,Z)}$ and $v'\in Z'$.  If $(q,Z)\in\cS^{-}$ then
    we get the result since \eqref{eq:invariant} holds before applying the Push
    rule.  Otherwise, we must have $(q',Z')=(q,Z)$ and we can choose the empty run
    $(q,v',\varepsilon)\xrightarrow{0}(q,v',\varepsilon)$.

  \smallskip\noindent\textbf{Pop rule.}
    Let $(q,Z)\in\cS=\cS^{-}$, $(q_2,Z_2)\in S_{(q,Z)}$ and $v'\in Z_2$. Again, if 
    $(q_2,Z_2)\in S^{-}_{(q,Z)}$ then we get the result since 
    \eqref{eq:invariant} holds before applying the Pop rule. Otherwise, by
    definition of the pop rule we have:
    \begin{enumerate}[nosep]
      \item some $(q',Z')\in S_{(q,Z)}$, 
      \item some push transition $t=q'\xrightarrow{g,\push_a,R}q''$,
      \item some $(q'',Z_1)\in \cS$ with $Z_1\sim_{q''}Z''=\overrightarrow{R(g\cap Z')}$,
      \item some $(q'_1,Z'_1)\in S_{(q'',Z_1)}$,
      \item some pop transition $t_1=q'_1\xrightarrow{g_1,\pop_a,R_1}q_2$,
    \end{enumerate}
    with $Z_2=\overrightarrow{R_1(g_1\cap Z'_1)}$. The construction below is 
    illustrated in Figure~\ref{fig:sound}.
    
    \begin{figure}[tbp]
      \centering
      $
      \begin{array}{ccccccccc}
        &  &  &  &  &  & (q'_1,v_4,a) & \xrightarrow{t_1}\xrightarrow{\delta'} & (q_2,v',\epsilon)  
        \\
        &  &  &  &  &  & \preceqv &  &   
        \\
        &  &  &  & (q'',v_3,a) & \xrightarrow{*} & (q'_1,v'_{4},a) &  & 
        \\
        &  &  &  & \preceqv &  &  &  & \preceqv
        \\
        &  & (q',v_{2},\varepsilon) & \xrightarrow{t}\xrightarrow{\delta} & (q'',v'_{3},a) &  & \preceqv &  &  
        \\
        &  & \preceqv &  & \preceqv  &  &  &  &   
        \\
        (q,v,\varepsilon) & \xrightarrow{*} & (q',v'_{2},\varepsilon) & 
        \xrightarrow{t}\xrightarrow{\delta} & (q'',v''_{3},a) & 
        \xrightarrow{*} & (q'_1,v''_{4},a) & 
        \xrightarrow{t_1}\xrightarrow{\delta'} & (q_2,v'',\epsilon) 
      \end{array}
      $
      \caption{Construction for the soundness.}
      \label{fig:sound}
    \end{figure}

    Since $v'\in Z_2$, we get some $v_4\in Z'_1$ such that $v_4\models g_1$ and 
    $v'=[R_1]v_4+\delta'$ for some $\delta'\geq0$. Hence we have a run
    $(q'_1,v_4,a)\xrightarrow{t_1}\xrightarrow{\delta'}(q_2,v',\varepsilon)$.
    
    Now, applying the invariant to $(q'',Z_1)\in\cS$,
    $(q'_1,Z'_1)\in S_{(q'',Z_1)}$ and $v_4\in Z'_1$,  we get a run 
    $(q'',v_3,\varepsilon)\xrightarrow{*}(q'_1,v'_4,\varepsilon)$ with 
    $v_3\in Z_1$ and $v_{4}\preceq_{q'_1}v'_{4}$. Hence, we also have a run
    $(q'',v_3,a)\xrightarrow{*}(q'',v'_{4},a)$.
    
    Let $v'_{3}\in Z''=\overrightarrow{R(g\cap Z')}\sim_{q''}Z_1$ with
    $v_3\preceq_{q''}v'_{3}$.  we get some $v_2\in Z'$ such that $v_{2}\models g$
    and $v'_{3}=[R]v_{2}+\delta$ for some $\delta\geq0$.  Hence we have a run
    $(q',v_{2},\varepsilon)\xrightarrow{t}\xrightarrow{\delta}(q'',v'_{3},a)$.
   
    Finally, we apply the invariant to $(q,Z)\in\cS$,
    $(q',Z')\in S_{(q,Z)}$ and  $v_{2}\in Z'$, we get a run 
    $(q,v,\varepsilon)\xrightarrow{*}(q',v'_{2},\varepsilon)$ with 
    $v\in Z$ and $v_{2}\preceq_{q'}v'_{2}$.

    By repeatedly applying the property of simulation $\preceq$,
    we may extend the run from $(q',v'_{2},\varepsilon)$ with
    $
    (q',v'_{2},\varepsilon) 
    \xrightarrow{t}\xrightarrow{\delta} (q'',v''_{3},a) 
    \xrightarrow{*} (q'_1,v''_{4},a) 
    \xrightarrow{t_1}\xrightarrow{\delta'} (q_2,v'',\varepsilon)
    $
    where $v_3\preceq_{q''}v'_{3}\preceq_{q''}v''_{3}$ and
    $v_{4}\preceq_{q'_1}v'_{4}\preceq_{q'_1}v''_{4}$.
    Finally $v'\preceq_{q_2}v''$.
    Therefore, the invariant holds after the pop rule.
    
\end{proof}

\section{Algorithm for PDTA Reachability via Zones}\label{sec:algo}

In this section, we describe Algorithm~\ref{algo:main} implementing the fixed point
computation defined by the inductive rules in Table~\ref{tbl:rules-tpda}.  We describe
the structure of the algorithm and its main data-structures.

Notice first that the sets $\cS$ and $S_{(q,Z)}$ for $(q,Z)\in\cS$ can be alternatively 
represented as a single set of pairs of nodes:
$$
\mathcal{S} = \{ [(q,Z),(q',Z')] \mid (q,Z)\in\cS \text{ and } (q',Z')\in S_{(q,Z)} \} \,.
$$
We can recover $\cS$ as the first projection of $\mathcal{S}$ and $S_{(q,Z)}$ as the 
second projection of $\mathcal{S}$ filtered by the first component being $(q,Z)$. We use 
both notations below depending on which is more convenient.
The start rule initializes $\mathcal{S}$ to $\{[(q_0,Z_0),(q_0,Z_0)]\}$.

Let us consider first the rule for internal transitions.  For each already discovered pair
of nodes $[(q,Z),(q',Z')]\in\mathcal{S}$ (or $(q',Z')\in S_{(q,Z)}$ with $(q,Z)\in\cS$),
we have to consider each possible internal transition $q'\xrightarrow{g,\nop,R}q''$ and
check whether the node $(q'',Z'')$ with $Z''=\overrightarrow{R(g\cap Z')}$ should be added
to $S_{(q,Z)}$ or is subsumed by an existing node.  This is like a graph traversal.  The
set $\mathcal{S}$ stores the already discovered pairs of nodes, and we will use a $\Todo$
(unordered) list to store the newly discovered nodes from which outgoing transitions
should be considered.  The $\Todo$ list should also consist of pairs $[(q,Z),(q',Z')]$ so
that when a new node $(q'',Z'')$ is discovered by an internal transition from $(q',Z')$ we
know to which set $S_{(q,Z)}$ it should be added.

As we can see from Theorem~\ref{thm:tpdarules}-soundness, given $(q,Z)\in\cS$, the set
$S_{(q,Z)}$ should consist of nodes reachable from $(q,Z)$ via a well-nested run.  Hence,
when dealing with a pair $[(q,Z),(q',Z')]\in\mathcal{S}$ and we see a push transition
$q'\xrightarrow{g,\push_a,R}q''$ with $Z''=\overrightarrow{R(g\cap Z')}$, we should not
try to add the pair $(q'',Z'')$ to $S_{(q,Z)}$ since the corresponding run would not be
well-nested.  Instead, we should search for a matching pop transition which could be taken
after a well-nested run starting from $(q'',Z'')$.  This is why the push rule adds the new
root $(q'',Z'')$ to $\cS$ (unless it is equivalent to an existing root).  The pair of
nodes $[(q'',Z''),(q'',Z'')]$ is newly discovered and added to the $\Todo$ list for
further exploration.  

The push transition may be matched with several pop transitions
(which could be already discovered or yet to be discovered by the algorithm).  To avoid
revisiting the push transition many times, it will be stored by the algorithm in an
additional set $\mathcal{S}_\push$.  More precisely, we will store in $\mathcal{S}_\push$
the tuple $[(q,Z),a,(q'',Z'')]$ meaning that the root node $(q'',Z'')$ may be reached from
the root node $(q,Z)$ via a well-nested run reaching some $(q',Z')$ followed by a 
transition pushing $a$ onto the stack.

Finally, assume that, when dealing with a pair $[(q_1,Z_1),(q'_1,Z'_1)]\in\mathcal{S}$, we
see a pop transition $q'_1\xrightarrow{g_1,\pop_a,R_1}q_2$ with
$Z_2=\overrightarrow{R_1(g_1\cap Z'_1)}$.  We will check whether it can be matched with an
already visited push transition, stored in the set $\mathcal{S}_\push$ as a pair
$[(q,Z),a,(q'',Z'')]$ with $(q'',Z'')=(q_1,Z_1)$.  If this is the case, the pop rule may
be applied and the node $(q_2,Z_2)$ added to $S_{(q,Z)}$ (unless it is subsumed by an
existing node).  The newly discovered pair of nodes $[(q,Z),(q_2,Z_2)]$ is also added to
the $\Todo$ list for further exploration. Once again, the pop transition may also be 
matched with push transitions that will be discovered later by the algorithm. To avoid 
revisiting the pop transition many times, we store the tuple $[(q_1,Z_1),a,(q_2,Z_2)]$ in 
a new set $\mathcal{S}_\pop$.

\medskip\noindent\textbf{Data structures.}\label{sec:data-structures}
We use a data structure $\TLM$ to store the triple of sets 
$(\mathcal{S},\mathcal{S}_\push,\mathcal{S}_\pop)$ and which is accessed with the 
following methods.
\begin{itemize}[nosep]
  \item  $\TLM.\create()$ creates the data structure with the three sets empty.

  \item  $\TLM.\add(q,Z,q',Z')$ adds $[(q,Z),(q',Z')]$ to $\mathcal{S}$.

  \item  $\TLM.\addpush(q,Z,a,q',Z')$ adds $[(q,Z),a,(q',Z')]$ to $\mathcal{S}_\push$.

  \item  $\TLM.\addpop(q,Z,a,q',Z')$ adds $[(q,Z),a,(q',Z')]$ to $\mathcal{S}_\pop$.

  \item  $\TLM.\isnewroot(q,Z)$ returns $[\false,Z']$ if there exists some 
  $(q,Z')\in\cS$ with $Z'\sim_q Z$, and returns $[\true,Z]$ otherwise.

  \item  $\TLM.\isnewnode(q,Z,q',Z')$ returns $\false$ if
  $\exists[(q,Z),(q',Z'')]\in\mathcal{S}$ with $Z'\preceq_{q'}Z''$, and returns $\true$ otherwise.

  \item  $\TLM.\isnewpop(q,Z,a,q',Z')$ returns $\false$ if 
  $\exists[(q,Z),a,(q',Z'')]\in\mathcal{S}_\pop$ with $Z'\preceq_{q'}Z''$, $\true$ otherwise.

  \item $\TLM.\isnewpush(q,Z,a,q',Z')$ returns $\false$ if
  $[(q,Z),a,(q',Z')]\in\mathcal{S}_\push$, and returns $\true$ otherwise.

  \item $\TLM.\iterpop(q,Z,a)$ returns the list of $(q',Z')$ with 
  $[(q,Z),a,(q',Z')]\in\mathcal{S}_\pop$.

  \item $\TLM.\iterpush(a,q',Z')$ returns the list of $(q,Z)$, s.t.
  \mbox{$[(q,Z),a,(q',Z')]\in\mathcal{S}_\push$}.
\end{itemize}

Concretely, the data structure should store sets of nodes $(q,Z)$ and be able to search 
or iterate through such sets. 
In order to make the algorithm slightly faster, we will segregate our sets of nodes, with
the name of the state.  We will use a hashmap in order to accomplish this task. See 
Figure~\ref{fig:TLM} where the concrete data structure is depicted.

\begin{algorithm}[tbp]
\footnotesize
\caption{PDTA Reachability Using Zones.}
\label{algo:main}
\begin{algorithmic}[1]
\Procedure{PDTAReach}{}
\State $\TLM.\create()$\label{line:startrule}
\State $\TLM.\add(q_0,Z_0,q_0,Z_0)$\label{line:startrule}\Comment{Start Rule}
\State $\Todo=\{[(q_0,Z_0),(q_0,Z_0)]\}$
\While{$\Todo \neq \emptyset $}
\State $[(q,Z),(q',Z')] = \Todo.\mathsf{get}()$\Comment{$(q,Z)\in\cS \wedge (q',Z')\in S_{(q,Z)}$}
\For{$t=q'\xrightarrow{g,\op,R}q''$ and $Z''=\overrightarrow{R(g\cap Z')}\neq\emptyset$}\label{line:forloop}
\If{$\op = \nop \land \TLM.\isnewnode(q,Z,q'',Z'')$}\label{line:nopisnewnode}
\State $\TLM.\add(q,Z,q'',Z'')$\label{line:internalrule}\Comment{Internal Rule}
\State $\Todo.\add([(q,Z),(q'',Z'')])$
\ElsIf{$\op = \push_a$}
\State $[isNew,Z_1] = \TLM.\isnewroot(q'',Z'')$\label{line:getoldroot}
\If{$isNew == \true$}\label{line:isNewRoot}
\State $\TLM.\add(q'',Z'',q'',Z'')$\label{line:pushrule}\Comment{Push Rule}
\State $\Todo.\add([(q'',Z''),(q'',Z'')])$
\EndIf
\If{$\TLM.\isnewpush(q,Z,a,q'',Z_1)$}\label{line:matchpopstart}
\State $\TLM.\addpush(q,Z,a,q'',Z_1)$\label{line:appPush}
\For{$(q_2,Z_2)$ in $\TLM.\iterpop(q'',Z_1,a)$}\label{line:formatchpopstart}
\If{$\TLM.\isnewnode(q,Z,q_2,Z_2)$}\label{line:ifmatchpop}
\State $\TLM.\add(q,Z,q_2,Z_2)$\label{line:poprule1}\Comment{Pop Rule}
\State $\Todo.\add([(q,Z),(q_2,Z_2)])$
\EndIf
\EndFor
\EndIf \label{line:matchpopend}
\ElsIf{$\op = \pop_a$}
\If{$\TLM.\isnewpop(q,Z,a,q'',Z'')$}\label{line:matchpushstart}
\State $\TLM.\addpop(q,Z,a,q'',Z'')$\label{line:appendpop}
\For{$(q_3,Z_3)$ in $\TLM.\iterpush(a,q,Z)$}\label{line:formatchpushstart}
\If{$\TLM.\isnewnode(q_3,Z_3,q'', Z'')$}\label{line:ifmatchpush}
\State $\TLM.\add(q_3,Z_3,q'',Z'')$\label{line:poprule2}\Comment{Pop Rule with $q=q_3,Z=Z_3$}
\State $\Todo.\add([(q_3,Z_3),(q'',Z'')])$\Comment{$q_2=q'',Z_2=Z''$}
\EndIf
\EndFor
\EndIf \label{line:matchpushend}
\EndIf
\EndFor
\EndWhile
\EndProcedure
\end{algorithmic}
\end{algorithm}

We will use a first level hashmap to store the set of roots $\cS$.  In order to implement $\TLM.\isnewnode(q,Z,q',Z')$, we first search for $(q,Z)$ in the first level map, then a pointer $\TLM[q][Z][0]$ will lead to a second level hashmap for the set of nodes $S_{(q,Z)}$ and we search for $(q',Z')$ in this second level map. 
See Figure~\ref{fig:TLM}(b).

To implement
$\TLM.\isnewpop(q,Z,a,q',Z')$ and $\TLM.\iterpop(q,Z,a)$, we
first search the root node $(q,Z)$ in the first level map, then a pointer $\TLM[q][Z][2]$
will lead to a second level hashmap storing the set of triples $(a,q',Z')$ such that
$[(q,Z),a,(q',Z')]\in\mathcal{S}_\pop$.  To speed up the access, this second level pop map
is segregated first on the key $a$, then on the key $q'$ to get the list of corresponding
zones $Z'$. See Figure~\ref{fig:TLM}(c,d).

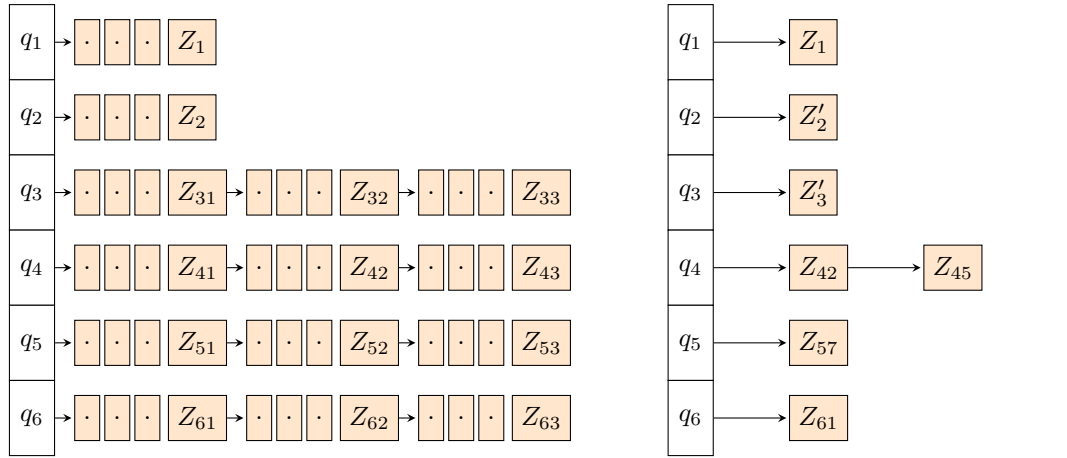
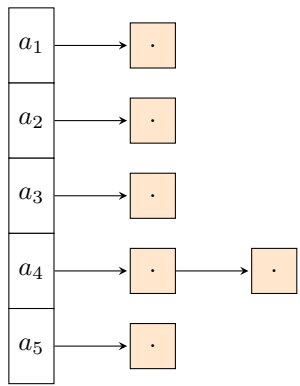
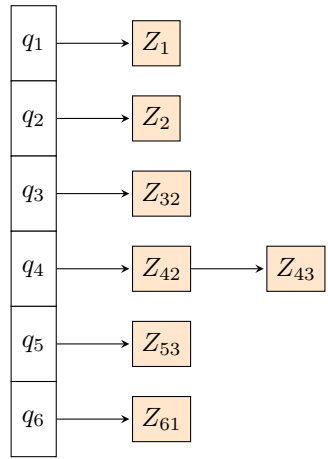
\begin{figure}[tbp]
    \centering
    \begin{subfigure}[t]{0.45\textwidth}
    \centering
    \begin{tikzpicture}
\foreach \index/\list in {1/{$Z_1$}, 2/{$Z_2$}, 3/{$Z_{31}$,$Z_{32}$,$Z_{33}$}, 4/{$Z_{41}$,$Z_{42}$,$Z_{43}$}, 5/{$Z_{51}$,$Z_{52}$,$Z_{53}$}, 6/{$Z_{61}$,$Z_{62}$,$Z_{63}$}} {
  \node[array element] (aux) at (0,-\index) {$q_{\index}$};
  \LinkedList{\list}
}
\end{tikzpicture}
    \caption{First level map constructed using equivalence $\sim_q$ for controlling size. 
Keys will be state names, values will be lists of quadruplets, each of which has four
pointers to second level maps, second level pushes maps, second level pops maps, and
zones.}
    \label{fig:label:a}
    \end{subfigure}
    \hfill
    \begin{subfigure}[t]{0.45\textwidth}
    \centering
    \begin{tikzpicture}
    \foreach \index/\list in {1/{$Z_1$}, 2/{$Z_2'$}, 3/{$Z_{3}'$}, 4/{$Z_{42}$,$Z_{45}$}, 5/{$Z_{57}$}, 6/{$Z_{61}$}} {
      \node[array element] (aux1) at (0,-\index) {$q_{\index}$};
      \LinkedListTwo{\list}
    }
    \end{tikzpicture}
    \caption{Second level map corresponding to $S_{(q_1,Z_1)}$. Each first level map node will have its own second level map.}
    \label{fig:label:b}
    \end{subfigure}
    \begin{subfigure}[t]{0.45\textwidth}
    \centering
    \begin{tikzpicture}
\foreach \index/\list in {1/{.},2/{.}, 3/{.}, 4/{.,.}, 5/{.}} {
  \node[array element] (aux1) at +
  (0,-\index) {$a_{\index}$};
  \LinkedListTwo{\list}
}
\end{tikzpicture}
\caption{Pushes/Pops map corresponding to root node $(q_1,Z_1)$.
Each pointer points to a different map where $(q,Z)$ are stored.}
\label{fig:label:a}
    \end{subfigure}
    \hfill
    \begin{subfigure}[t]{0.45\textwidth}
    \centering
    \begin{tikzpicture}
    \foreach \index/\list in {1/{$Z_1$}, 2/{$Z_2$}, 3/{$Z_{32}$}, 4/{$Z_{42}$,$Z_{43}$}, 5/{$Z_{53}$}, 6/{$Z_{61}$}} {
      \node[array element] (aux1) at +
      (0,-\index) {$q_{\index}$};
      \LinkedListTwo{\list}
    }
    \end{tikzpicture}
    \caption{For pushes/pops map, this is a map corresponding to root node $(q_1,Z_1)$,
and symbol $a_2$ (say).  The $(q,Z)$ stored here is constructed using equivalence
(pushes map), or using simulation (pops map).}
\label{fig:label:b}
    \end{subfigure}
    \caption{Two level map implementing the data structure $\TLM$ storing the sets 
$\mathcal{S}$, $\mathcal{S}_\push$, $\mathcal{S}_\pop$.}
    \label{fig:TLM}
\end{figure}

Finally, we also store the set $\mathcal{S}_\push$ to implement
$\TLM.\isnewpush(q,Z,a,q',Z')$ and $\TLM.\iterpush(a,q',Z')$.  Notice that
$\mathcal{S}_\push$ consists of triples $[(q,Z),a,(q',Z')]$ where both $(q,Z)$ and
$(q',Z')$ are root nodes from $\cS$.  Notice also that for the iteration we fix the second
node $(q',Z')$.  To get an efficient implementation, we first search the root node
$(q',Z')$ in the first level map, then a pointer $\TLM[q'][Z'][1]$ will lead to a second
level hashmap storing the set of triples $(a,q,Z)$ such that
$[(q,Z),a,(q',Z')]\in\mathcal{S}_\push$.  To speed up the access, this second level push
map is segregated first on the key $a$, then on the key $q$ to get the list of
corresponding zones $Z$.  See Figure~\ref{fig:TLM}(c,d).

\medskip
We move next to proving the main lemmas that will help us show correctness of 
Algorithm~\ref{algo:main}. We start with termination.

\begin{lemma}[Termination]\label{lem:term_lemma}
  Algorithm~\ref{algo:main} will always terminate.
\end{lemma}
\begin{proof}
  If we look at our algorithm, we can clearly see that before adding a pair of nodes to
  the $\Todo$ list, we add the same pair to $\mathcal{S}$ with $\TLM.\add$, and before
  that, we always check whether the pair is already in our $\TLM$ or not ($\isnewnode$ or
  $\isnewroot$).  Since the size of the $\TLM$ is always bounded because we check either
  the first level map or the second level map before adding, the outer while loop will be
  called only a finite number of times, which shows termination of the algorithm. \end{proof}

Next, recall that $\TLM$ encodes a triple of sets 
$(\mathcal{S},\mathcal{S}_\push,\mathcal{S}_\pop)$ defined by:
\begin{align*}
  \mathcal{S} & =\{[(q,Z),(q',Z')] \mid (q',Z')\in\TLM[q][Z][0] \} \\
  \mathcal{S}_\push & =\{[(q,Z),a,(q',Z')] \mid (a,q,Z)\in\TLM[q'][Z'][1] \} \\
  \mathcal{S}_\pop & =\{[(q,Z),a,(q',Z')] \mid (a,q',Z')\in\TLM[q][Z][2] \} 
\end{align*}
Recall also the correspondence explained at the beginning of Section~\ref{sec:algo}
between a set $\mathcal{S}$ of pairs of nodes, and the set of roots $\cS$ together with 
the sets of nodes $S_{(q,Z)}$ for $(q,Z)\in\cS$. Then we have the following properties.

\begin{lemma}\label{lem:todotlm}
  If $[(q,Z),(q',Z')] \in \Todo$, then $(q',Z')\in \TLM[q][Z][0]$.
\end{lemma}

\begin{proof}
  $\TLM[q][Z][0]$ is changed only when $\add$ method is called.  As we can see in
  the algorithm, every addition to $\Todo$ with
  $\Todo.\add([(q,Z),(q',Z')])$
  follows a call 
  $\TLM.\add(q,Z,q',Z')$
  resulting in $(q',Z')\in \TLM[q][Z][0]$.
  And since $\TLM$ only increases, we can say that the lemma is true.  The
  converse is not true, as $\Todo$ can increase and decrease.
  
\end{proof}

\begin{lemma}[Pops]\label{lem:pop_constraint}
  If $(a,q'',Z'')\in\TLM[q][Z][2]$, then there exist $(q',Z')\in\TLM[q][Z][0]$
  and some transition $t=q'\xrightarrow{g,\pop_a,R}q''$ with
  $Z''=\overrightarrow{R(g\cap Z')}\neq\emptyset$.
\end{lemma}

\begin{proof}
  The method $\addpop$ is only used in Line~\ref{line:appendpop}, where $(a,q'',Z'')$ is
  added to $\TLM[q][Z][2]$.  When reaching this line, the pair extracted from $\Todo$ is
  $[(q,Z),(q',Z')]$ and by Lemma~\ref{lem:todotlm}, we get $(q',Z')\in\TLM[q][Z][0]$.  The
  for loop on Line~\ref{line:forloop} considers a transition
  $t=q'\xrightarrow{g,\op,R}q''$ such that $Z''=\overrightarrow{R(g\cap
  Z')}\neq\emptyset$.  Moreover, we have $\op=\pop_a$, which completes the proof.  
  
\end{proof}

\begin{lemma}[Pushes]\label{lem:push_constraint}
  If $(a,q,Z)\in\TLM[q''][Z_1][1]$, then there exist $(q',Z')\in\TLM[q][Z][0]$ and
  some transition $t=q'\xrightarrow{g,\push_a,R}q''$ such that
  $Z''=\overrightarrow{R(g\cap Z')}\neq\emptyset$ and $Z''\sim_{q''}Z_1$.  
\end{lemma}

\begin{proof}
  The method $\addpush$ is only used in Line~\ref{line:appPush}, when we add $(a,q,Z)$
  into $\TLM[q''][Z_1][1]$.  We have $Z_1\sim_{q''}Z''$ due to the call to $\isnewroot$ 
  on Line~\ref{line:getoldroot}.
  When reaching this line, the pair extracted from $\Todo$ is $[(q,Z),(q',Z')]$ and by
  Lemma~\ref{lem:todotlm}, we get $(q',Z')\in\TLM[q][Z][0]$.  The for loop on
  Line~\ref{line:forloop} considers a transition $t=q'\xrightarrow{g,\op,R}q''$ such that
  $Z''=\overrightarrow{R(g\cap Z')}\neq\emptyset$.  Moreover, we have $\op=\push_a$, which
  completes the proof.  
  
\end{proof}

Using the three properties above, we will now show soundness and completeness of Algorithm~\ref{algo:main}. We start with soundness, for which we need to define the notion of a mid-computation.
\begin{definition}[Mid-Computation (MC)]\label{def:mcdef}
  A mid-computation is any set $\mathcal{S}$ of pairs of nodes $[(q,Z),(q',Z')]$ 
  that can be obtained by applying the inductive rules in Table~\ref{tbl:rules-tpda} in
  some order (not necessarily until a fixed point is reached).  
  We say that $\TLM$ is a mid-computation when it encodes a triple 
  $(\mathcal{S},\mathcal{S}_\push,\mathcal{S}_\pop)$ where $\mathcal{S}$ is a 
  mid-computation.
\end{definition}

For example, just applying the start rule yields $\mathcal{S}=\{[(q_0,Z_0),(q_0,Z_0)]\}$,
which is a mid-computation. Now we can show the following lemma. 

\begin{lemma}[Soundness]\label{lem:algo-sound}
  At any point in Algorithm~\ref{algo:main}, 
  $\TLM$ is a mid-computation.
\end{lemma}
\begin{proof}
  We will prove this lemma by induction.
  
  \noindent\textbf{Base Case}: Just after execution of Line~\ref{line:startrule}, we have 
  $\mathcal{S} = \{[(q_0,Z_0),(q_0,Z_0)]\}$
  which is the smallest MC, just after applying the start rule.
  
  \noindent\textbf{Induction Hypothesis}: Let us say that after letting the algorithm
  execute for some time we stop it, and $\TLM$ is a mid-computation.
  
  \noindent\textbf{Induction}: Now we prove that no matter what happens, $\TLM$ will still
  be a MC. There are only a few statements in the algorithm that will change $\mathcal{S}$
  encoded by $\TLM$, and that is when some $\add$ method is called.  Apart from the
  initialization in Line~\ref{line:startrule} which has been considered in the base case,
  they all start with a pair $[(q,Z),(q',Z')]$ popped from $\Todo$ list.  According to
  Lemma~\ref{lem:todotlm}, we can see that $(q',Z')\in \TLM[q][Z][0]$, i.e.,
  $[(q,Z),(q',Z')]\in\mathcal{S}$ or equivalently $(q,Z)\in\cS$ and $(q',Z')\in S_{(q,Z)}$.  
  Now, we have several cases.
  \begin{itemize}[nosep]
    \item $\TLM.\add(q,Z,q'',Z'')$ on Line~\ref{line:internalrule}: 
    The transtion in consideration is $t=q'\xrightarrow{g,\nop,R}q''$ with
    $Z''=\overrightarrow{R(g\cap Z')}\neq\emptyset$.
    From the test $\TLM.\isnewnode(q,Z,q'',Z'')$, we know that $(q'',Z'')$
    is not already subsumed by an element in $S_{(q,Z)}$.  Therefore, the rule for 
    internal transitions
    can be applied, resulting in $\mathcal{S}:=\mathcal{S}\cup\{[(q,Z),(q'',Z'')]\}$.  And
    this corresponds to excecuting Line~\ref{line:internalrule}.  Therefore, the
    new $\TLM$ after Line~\ref{line:internalrule} is also a mid-computation.
    
    \item $\TLM.\add(q'',Z'',q'',Z'')$ on Line~\ref{line:pushrule}: Now,
    $t=q'\xrightarrow{g,\push_a,R}q''$ is a push transition with
    $Z''=\overrightarrow{R(g\cap Z')}\neq\emptyset$.
    The call to $\isnewroot$ on Line~\ref{line:getoldroot} returned $[true,Z'']$, indicating
    that $(q'',Z'')$ is a new root. Therefore, the push rule can be applied here, which is what
    is done in Line~\ref{line:pushrule} and $\TLM$ is still a mid-computation.
    
    \item $\TLM.\add(q,Z,q_2,Z_2)$ on Line~\ref{line:poprule1}: 
    The transition in consideration is $t=q'\xrightarrow{g,\push_a,R}q''$
    with $Z''=\overrightarrow{R(g\cap Z')}\neq\emptyset$.
    According to Line~\ref{line:getoldroot}, $Z_1\sim_{q''}Z''$.  Since we are looping on
    $\TLM.\iterpop(q'',Z_1,a)=\TLM[q''][Z_1][2][a]$, by Lemma~\ref{lem:pop_constraint},
    there exist a node $(q'_1,Z'_1)\in\TLM[q''][Z_1][0]$ and a transition
    $t_1=q'_1\xrightarrow{g_1,\pop_a,R_1}q_2$ with
    $Z_2=\overrightarrow{R_1(g_1\cap Z'_1)}\neq\emptyset$.
    Since $\TLM$ encodes $\mathcal{S}$, we have $[(q'',Z_1),(q'_1,Z'_1)]\in\mathcal{S}$, 
    i.e., $(q'',Z_1)\in\cS$ and $(q'_1,Z'_1)\in S_{(q'',Z_1)}$.  We can therefore apply the
    pop rule over here, unless $(q_2,Z_2)$ is subsumed by an existing node in $S_{(q,Z)}$,
    which is checked in Line~\ref{line:ifmatchpop}.
    Therefore the pop rule can be applied, which is exactly done in 
    Line~\ref{line:poprule1}, resulting in $\TLM$ being again a mid-computation.
    
    \item $\TLM.\add(q_3,Z_3,q'',Z'')$ on Line~\ref{line:poprule2}: 
    The main transition under consideration is $t=q'\xrightarrow{g,\pop_a,R}q''$
    with $Z''=\overrightarrow{R(g\cap Z')}\neq\emptyset$.
    We are looping on $(q_3,Z_3)$ in $\TLM.\iterpush(a,q,Z)=\TLM[q][Z][1][a]$.  Hence, by
    Lemma~\ref{lem:push_constraint}, that there exist $(q'_3,Z'_3)\in\TLM[q_3][Z_3][0]$
    and some transition $t_3=q'_3\xrightarrow{g_3,\push_a,R_3}q$ such that
    $Z''_3=\overrightarrow{R_3(g_3\cap Z'_3)}\sim_q Z$.
    Since $\TLM$ encodes $\mathcal{S}$, we have $(q_3,Z_3)\in\cS$ and $(q'_3,Z'_3)\in
    S_{(q_3,Z_3)}$.  We can therefore apply the pop rule here as well, if the node
    $(q'',Z'')$ is not subsumed by any other node in $S_{(q_3,Z_3)}$, which is checked in
    Line~\ref{line:ifmatchpush}.  The pop rule can therefore be applied and the resulting
    $\TLM$ is still a mid-computation.
  \end{itemize}
\end{proof}

Next we turn to completeness of the algorithm.
\begin{lemma}[Completeness]\label{lem:algo-complete}
  After the termination of Algorithm~\ref{algo:main}, the set $\mathcal{S}$ encoded by 
  $\TLM$ is 
  a fixed point of the inductive rules in Table~\ref{tbl:rules-tpda}, that is, no new
  node can be added by applying one of these rules.
\end{lemma}
\begin{proof}
We show that the rules defined in Table~\ref{tbl:rules-tpda} will not add any new 
pair of nodes to the set $\mathcal{S}$ defined by $\TLM$.
\begin{itemize}[nosep]
  \item \textbf{Start Rule}: After Line~\ref{line:startrule} of the algorithm we 
  see that $[(q_0,Z_0),(q_0,Z_0)]\in\mathcal{S}$.  Therefore, applying the start rule would not
  add any new pair of nodes.

  \item \textbf{Internal Rule}: Let $(q,Z)\in\cS$ and $(q',Z')\in S_{(q,Z)}$, i.e.,
  $[(q,Z),(q',Z')]\in\mathcal{S}$.  Let $t=q'\xrightarrow{g,\nop,R}q''$ be an internal
  transition and $Z''=\overrightarrow{R(g\cap Z')}\neq\emptyset$.  Applying the internal
  rule in Table~\ref{tbl:rules-tpda} would add $(q'',Z'')$ to $S_{(q,Z)}$ unless there is
  $(q'',Z''')\in S_{(q,Z)}$ with $Z''\preceq_{q''}Z'''$.
  Now, since $[(q,Z),(q',Z')]\in\mathcal{S}$ then $\TLM.\add(q,Z,q',Z')$ was executed,
  and $[(q,Z),(q',Z')]$ was also added to $\Todo$.  When $[(q,Z),(q',Z')]$ was removed from
  $\Todo$, the internal transition $t$ was considered.  Hence, either the test
  $\TLM.\isnewnode(q,Z,q'',Z'')$ returned true and $\TLM.\add(q,Z,q',Z')$ was executed; we
  get $(q'',Z'')\in S_{(q,Z)}$.  Or the test returned false and we had some
  $(q'',Z''')\in\TLM[q][Z][0]=S_{(q,Z)}$ such that $Z''\preceq_{q''}Z'''$.  In both cases,
  the rule for internal transitions is satisfied.
  
  \item \textbf{Push Rule}: Let $(q,Z)\in\cS$ and $(q',Z')\in S_{(q,Z)}$, i.e.,
  $[(q,Z),(q',Z')]\in\mathcal{S}$.  Let $t=q'\xrightarrow{g,\push_a,R}q''$ be a push
  transition and $Z''=\overrightarrow{R(g\cap Z')}\neq\emptyset$.  Since
  $[(q,Z),(q',Z')]\in\mathcal{S}$ then $\TLM.\add(q,Z,q',Z')$ was executed, and
  $[(q,Z),(q',Z')]$ was also added to $\Todo$.  When $[(q,Z),(q',Z')]$ was removed from
  $\Todo$, the push transition $t$ was considered and $\TLM.\isnewroot(q'',Z'')$ was
  called and returned $[isNew,Z_1]$.  If the boolean $isNew$ is true then
  Line~\ref{line:pushrule} is executed resulting in $[(q'',Z''),(q'',Z'')]$, i.e.,
  $(q'',Z'')\in\cS$ and $(q'',Z'')\in S_{(q'',Z'')}$.  Otherwise, $(q'',Z_1)\in\cS$ and
  $Z''\sim_{q''}Z_1$.  In both cases, the rule for push transitions is satisfied.
  
  \item \textbf{Pop Rule}: Let the rule premises be $(q,Z) \in \cS$, $(q',Z')\in
  S_{(q,Z)}$, $t=q'\xrightarrow{g,\push_a,R}q''$, $Z''=\overrightarrow{R(g\cap Z')}\sim
  Z_1$; and $(q'',Z_1)\in\cS$, $(q_1',Z_1')\in S_{(q'',Z_1)}$,
  $t_1=q'_1\xrightarrow{g_1,\pop_a,R_1}q_2$, $Z_2=\overrightarrow{R_1(g_1\cap
  Z'_1)}\neq\emptyset$.  We will show that there is $(q_2,Z_2^{\dagger})\in S_{(q,Z)}$
  with $Z_2\preceq_{q_2}Z_2^{\dagger}$ so that the conclusion of the pop rule is
  satisfied.
  
  First, by the argument used earlier with respect to $[(q,Z),(q',Z')]\in\mathcal{S}$, we
  can say that the transition $t$, hence also the node $(q'',Z'')$, should have been
  considered while popping $[(q,Z),(q',Z')]$ from the $\Todo$ list.  Since no two nodes in
  the first level map of $\TLM$ are equivalent, we can say that $Z_1$ used in
  Line~\ref{line:matchpopstart}-\ref{line:matchpopend} coincide with $Z_1$ used here.  Due
  to Lines~\ref{line:matchpopstart}-\ref{line:appPush}, we get
  $(a,q,Z)\in\TLM[q''][Z_1][1]$.  Let us note $T_\push$ the time when 
  $\TLM.\addpush(q,Z,a,q'',Z_1)$ was executed.
  Note that this time could have been earlier than when
  transition $t$ was considered if condition on Line~\ref{line:matchpopstart} was not
  satisfied.
  
  With respect to $(q'_1,Z'_1)\in S_{(q'',Z_1)}$, i.e.,
  $[(q'',Z_1),(q'_1,Z'_1)]\in\mathcal{S}$, we can say that transition $t_1$, hence also
  node $(q_2,Z_2)$, should have been considered when $[(q'',Z_1),(q'_1,Z'_1)]$ was removed
  from the $\Todo$ list.  Notice the correspondence between the notation in the algorithm
  and the notation above, correspondence given by:
  \begin{align*}
    q & \mapsto q'' & Z &\mapsto Z_1 & q' & \mapsto q'_1 & Z' &\mapsto Z'_1 \\
    q'' & \mapsto q_2 & Z'' &\mapsto Z_2 & q_3 & \mapsto q & Z_3 &\mapsto Z 
  \end{align*}
  Hence, the test on Line~\ref{line:matchpushstart} translates to
  $\TLM.\isnewnode(q'',Z_1,a,q_2,Z_2)$.  If it returns false, there must be some
  other node $(q_2,Z_2')\in\TLM[q''][Z_1][2][a]$ with $Z_2\preceq_{q_2}Z_2'$.  If
  it returns true then $\TLM.\addpop(q'',Z_1,a,q_2,Z_2)$ is executed, resulting in
  $(q_2,Z_2)\in\TLM[q''][Z_1][2][a]$, and to unify the notation, we let $Z'_2=Z_2$.
  Let us note  $T_\pop$ the time when $\TLM.\addpop(q'',Z_1,a,q_2,Z'_2)$ was executed.
    
  We have the following cases now
  \begin{itemize}
    \item $T_\push > T_\pop$: In this case, the for loop Line~\ref{line:formatchpopstart}
    is executed with node $(q_2,Z_2')$ in $\TLM.\iterpop(q'',Z_1,a)=\TLM[q''][Z_1][2][a]$.
    Due to Lines~\ref{line:ifmatchpop}-\ref{line:poprule1}, we get either $(q_2,Z'_2)\in
    \TLM[q][Z][0]=S_{(q,Z)}$ or there is some $(q_2,Z_2^{\dagger})\in
    \TLM[q][Z][0]=S_{(q,Z)}$ such that $Z'_2\preceq_{q_2}Z_2^{\dagger}$.  In both cases,
    using $Z_2\preceq_{q_2}Z'_2$, we deduce that the pop rule is satisfied.
    
    \item $T_\push < T_\pop$: In this case, the for loop Line~\ref{line:formatchpushstart}
    is executed with node $(q,Z)$ in $\TLM.\iterpush(a,q'',Z_1)=\TLM[q''][Z_1][1][a]$
    (recall the correspondence of notations).  Due to
    Lines~\ref{line:ifmatchpush}-\ref{line:poprule2}, we get either $(q_2,Z_2)\in
    \TLM[q][Z][0]=S_{(q,Z)}$ or there is some $(q_2,Z_2^{\dagger})\in
    \TLM[q][Z][0]=S_{(q,Z)}$ such that $Z_2\preceq_{q_2}Z_2^{\dagger}$.  In both cases, we
    deduce that the pop rule is satisfied.  
  \end{itemize}
\end{itemize}
\end{proof}

Finally, from Lemma~\ref{lem:term_lemma}, Lemma~\ref{lem:algo-sound}, 
Lemma~\ref{lem:algo-complete} and Theorem~\ref{thm:tpdarules}, we obtain:

\begin{theorem}\label{thm:main-algo}
  The set $\mathcal{S}$ encoded by $\TLM$ computed by Algorithm~\ref{algo:main} is a fixed
  point obtained starting from the empty set by applying the inductive rules in
  Table~\ref{tbl:rules-tpda}.  Therefore, Algorithm~\ref{algo:main} terminates and is
  sound and complete for well-nested control state reachability of pushdown timed
  automata.
\end{theorem}

\section{Experiments and Results}
\label{sec:experiments}
We build on the existing architecture of an open-source tool for analysis of timed automata, TChecker~\cite{tchecker}.  Our tool is freely available for download and use at \url{https://github.com/karthik-314/PDTA_reachability}. Currently, our implementation uses only the LU-simulation~\cite{BBLP04} and hence we only deal with diagonal-free timed automata. But we could also extend it to use $\mathcal{G}$-simulations from~\cite{GMS19}, which would allow us to handle diagonal constraints as well.

\subsection{Implementation}
The input for our implementation are PDTA, rather than TA so we modify TChecker in order to run our experiments.  While most of the TChecker file format will remain the same, the only place where we make a change to the syntax of the input, will be the edges.  TChecker uses the following format, for its transitions,
\begin{verbatim}
    edge:<Process>:<src>:<tgt>:<label>{
        do:<Reset1(x=0)> ; <Reset2(y=0)> :
        provided: <guard1(x==0)> && <guard2(y>=1)>}
\end{verbatim}
The new format in order to incorporate  the pushes and pops will be,
\begin{verbatim}
    edge:<Process>:<src>:<tgt>:<label>{
        do:<Reset1(x=0)> ; <Reset2(y=0)> :
        provided: <guard1(x==0)> && <guard2(y>=1)>}
        [<push/pop>:<symbol>]
\end{verbatim}
In case the operation is $\nop$, the square brackets are left empty.

We have implemented two variants of Algorithm~\ref{algo:main} for PDTA and we will compare these between each other and also with a region-based approach.  More precisely, we consider the following 3 algorithms:
\begin{itemize}[nosep]
  \item \textbf{Simulation Based Approach} ($\preceq_{LU}$): Direct implementation of
  Algorithm~\ref{algo:main}.

  \item \textbf{Equivalence Based Approach} ($\sim_{LU}$): This is a variation of
  Algorithm~\ref{algo:main}, with two methods changed,
    \begin{itemize}
      \item \verb|TLM.isNewNode|$(q,Z,q',Z')$: Returns \verb|false| if
      $\exists[(q,Z),(q',Z'')] \in \mathcal{S}$ with $Z' \sim_{q'} Z''$, and \verb|true|
      otherwise.

      \item \verb|TLM.isNewPop|$(q,Z,a,q',Z')$: Returns \verb|false| if
      $\exists[(q,Z),a,(q',Z'')] \in \mathcal{S}_{pop}$ with $Z' \sim_{q'} Z''$, and
      \verb|true| otherwise.
      
    \end{itemize}
  As mentioned in
  Section~\ref{sec:rewrite-rules}, if instead of simulation, we just use equivalence
  everywhere, we do obtain a correct algorithm for reachability in PDTA. Hence it is
  interesting to compare it with the above approach.  
  
  \item \textbf{Region Based Implementation} ($RB$): A previous
  implementation~\cite{AkshayGKS17}, uses a region based approach in order to solve the
  non-emptiness problem in PDTA. We note two features of the algorithm.  First, it uses a
  tree-automaton based approach for efficiency and correctness, but underlying it is the
  region (rather than zone) construction.  Second, it works only with closed guards, while
  our approach works with closed and open guards.
\end{itemize}

We note the following important points regarding our implementation:
\begin{enumerate}[nosep]
    \item The $\preceq$ used in our implementation will be $\preceq_{LU}$~\cite{BBLP04},
    without extrapolation and with global clock bounds.

    \item The \verb|ToDo| list used currently uses LIFO (stack) ordering for popping of
    elements.  This corresponds to a DFS exploration of the zone-graph.  But we can use
    other data structures for this purpose as well, e.g., changing it to FIFO would give
    us a BFS exploration etc.

    \item Both the simulation based and equivalence based approach are tested on PDTA with
    empty and non-empty languages, but we have ensured that both of them return an answer
    only after the entire exploration has been completed.  In other words, we do not stop
    the exploration when we reach a final state.  This is to make fair comparisons, where
    we do not terminate because of being ``lucky'' in encountering the final state early.
    Of course in practice we would not do this.  In contrast, we note that the $RB$
    approach is an on the fly approach which returns non-empty as soon as the final state
    turns out to be reachable.
\end{enumerate}

All experiments are run on Intel-i5 10th Generation processor, with an 8GB RAM, with a timeout of $120$ seconds.
\subsection{Benchmarks} We used a total of 10 benchmarks in our experiments, but
parameterized several of them in order to test the scalability and to give us more insight
into performance comparisons.  The benchmark and their parameterizations are
explained in~\ref{app:benchmarks}.  We highlight only
some salient points here.
The benchmark $B_1$ is the PDTA from Figure~\ref{fig:tpda}.  $B_2(k)$ is directly adapted
from Figure~\ref{fig:tpda2} with the constant $y\leq 1$ parametrized to $y\leq k$, and
$k+1$ pops between $q_0$ and $q_2$.  Note that $q_3$ is unreachable regardless of the
value of $k$.  Benchmarks $B_3,B_4$ are adapted from~\cite{AkshayGKS17} with $B_3$
involving untiming of a stack age into normal clocks.  $B_5,B_6$ involve significant
interplay of push/pop edges and clocks and $B_6,B_7$ also have open guards.  More details
can be found in~\ref{app:benchmarks}.  We also note that automata $B_1$, $B_3(3, 4),
B_5(k_1,k_2), B_8, B_9(k_1,k_2)$ accept a nonempty language, while the rest are empty.  As
described earlier this does not change the performance of the simulation and equivalence
based approaches, but may significantly change the performance of the Region Based
Approach.
\begin{table}[t]
    \centering
      \begin{tabular}{ |p{2.3cm}||p{1.2cm}|p{1.4cm}|p{1.3cm}|p{1.4cm}|p{1.2cm}|p{1.2cm}|  }
 \hline
  Benchmark& \quad $\preceq_{LU}$ & \quad $\preceq_{LU}$ & \quad $\sim_{LU}$ &\quad $\sim_{LU}$& \quad $RB$ & \quad $RB$\\
 & Time & \# nodes & Time & \# nodes& Time & \# nodes\\
 \hline
 $B_1$&   0.2  & 17& 0.2 &17 &235.6 &4100\\
 $B_2(10)$   & 0.8    &77&   0.8 &77 & 6835.8 &30200\\
 $B_2(100)$&   20.0  & 5252   &20.7 &5252 & T.O. & $\geq$154700\\
 $B_3(4,3)$& 0.2  & 6   & 0.2 & 6 & 1043.8 & 14300\\
 $B_3(3,4)$& 0.2  & 9 &0.2 &9 & 98.8 & 3400\\
 $B_4$    &0.2 & 8&  0.1 & 8&0.3 &17 \\
 $B_5(100,10)$& 0.8 & 202 & 5.4 &2212 &OoM &OoM\\
 $B_5(100,1000)$& 0.7 & 202 & 3564.3 &201202 &OoM &OoM\\
 $B_5(5000,100)$& 23.2& 10002 &3429.3 &1010102 &OoM &OoM\\
 $B_6(5,4,1000)$& 0.3 & 30 &611.8 &30047 &NA &NA\\
 $B_6(5,4,10000)$& 0.3 & 30 &60271.9 &300047 &NA &NA\\
 $B_6(501,500,100)$& 38.2 & 3006 &501.0 &34799 &NA &NA\\
 $B_{7}$& 112.4 & 4475 & 113.1 &4475 &NA &NA\\
 \hline
\end{tabular}
    \caption{Results on the Benchmarks.  Time recorded in $ms$, and timeout
    (T.O.) used is 120 seconds.  OoM stands for Out of Memory kill.
    Results rounded up to 1 decimal.  \# nodes refers to the number of nodes in the
    zone/region graph explored.  In case of timeout $\geq n$, refers to recorded number of
    nodes $n$ before timeout occurred.  NA in $RB$ columns represents that the region
    based approach does not handle open guards in transitions ($B_6$, $B_7$ have open
    guards.)}
    \label{tab:final_results}
\end{table}

\subsection{Results}
Table~\ref{tab:final_results} contains a selection of our experimental results; more detailed results can
be found in~\ref{app:all_results}.  From the table, we
conclude first that the zone based approach is indeed faster than the Region Based
Approach for all examples.  Second, the simulation based approach runs faster than the
equivalence based approach for all examples if the \verb|ToDo| priority for removal
remains the same.  In fact, the performance of the simulation based approach depends
mostly on the size of the PDTA, but the equivalence based approach is dependant on the
constants used in guards as well, which is even more the case for the region based approach.  Finally, our approach can easily handle closed and open guards.

Most of the timeouts that occurred during the experiments are due to Out of Memory
(OoM) kills, especially for larger sized PDTAs.  For smaller sized PDTA such as
$B_2(100)$, the recorded number of nodes before timeout was 154700.

Regarding the performance, we would like to emphasize that B1, B2, B3, B4, B7 were
designed to compare the Zone approach to the region (RB) approach.  As a consequence these
models are very simple and the number of explored nodes remains almost the same regardless
of whether we use $\sim$ or $\preceq$ to prune, which reflects in the times/sizes not
being too different.  However, the other examples B5, B6 are more complex and have nodes
that get pruned during exploration (both using $\sim$ and $\preceq$).  Here we can see the
clear improvement of $\preceq$ over $\sim$ both in terms of time taken and also of number
of explored nodes.

\section{Discussion and Future work}
\label{sec:conclusion}

In this paper, we examined how an unbounded stack can be integrated seamlessly with
zone-abstractions in timed automata.  We would like to point out that two easy extensions
of our work are possible.  First, as remarked earlier, our algorithm checks for
well-nested reachability, i.e., it requires to reach a final state with empty stack for
acceptance.  But we can generalize this to general control-state reachability by showing
that a control state $q$ is reachable in the PDTA (with possibly a non-empty stack) iff
some node $(q,Z)$ is discovered by our algorithm and added to some $S_{(q',Z')}$ (and not
just to $S_{(q_0,Z_0)}$ as in the well-nested case).  While this idea is simple and
requires only minor edits to the existing algorithm, the proof of correctness requires
more work and we leave this for future work.

Secondly, we can handle the model with ages in stack as in~\cite{AGK16,AbdullaAS12} with
an exponential blowup (thanks to~\cite{ClementeL15}). However, an open question is
whether this blowup can be avoided in practice.  As noted earlier, there exist
extensions~\cite{CLLM17,CL21} studied especially in the context of binary reachability,
which are expressively strictly more powerful, for which decidability results are known.
It would be interesting to see how we can extend the zone-based approach to those models.

Finally, it seems interesting to examine further the link to the liveness problem,
possibly allowing us to transfer ideas and obtain faster implementations.  Another
possibility would be to use the extrapolation operator (rather than, or in addition to,
simulation), which we have not considered in this work.

\bibliographystyle{plain}

\bibliography{main_cav_arxiv_lipics_bib}

\begin{thebibliography}{10}

\bibitem{AbdullaAS12}
Parosh~Aziz Abdulla, Mohamed~Faouzi Atig, and Jari Stenman.
\newblock Dense-timed pushdown automata.
\newblock In {\em Proceedings of the 27th Annual {IEEE} Symposium on Logic in
  Computer Science, {LICS} 2012, Dubrovnik, Croatia, June 25-28, 2012}, page
  35–44, 2012.

\bibitem{AGJK19}
S.~Akshay, Paul Gastin, Vincent Jug{\'{e}}, and Shankara~Narayanan Krishna.
\newblock Timed systems through the lens of logic.
\newblock In {\em 34th Annual {ACM/IEEE} Symposium on Logic in Computer
  Science, {LICS} 2019, Vancouver, BC, Canada, June 24-27, 2019}, pages 1--13,
  2019.

\bibitem{AGK16}
S.~Akshay, Paul Gastin, and Shankara~Narayanan Krishna.
\newblock {Analyzing Timed Systems Using Tree Automata}.
\newblock {\em {Logical Methods in Computer Science}}, {Volume 14, Issue 2},
  May 2018.

\bibitem{AGKR20}
S.~Akshay, Paul Gastin, Shankara~Narayanan Krishna, and Sparsa Roychowdhury.
\newblock Revisiting underapproximate reachability for multipushdown systems.
\newblock In {\em Tools and Algorithms for the Construction and Analysis of
  Systems - 26th International Conference, {TACAS} 2020, Held as Part of the
  European Joint Conferences on Theory and Practice of Software, {ETAPS} 2020,
  Dublin, Ireland, April 25-30, 2020, Proceedings, Part {I}}, volume 12078 of
  {\em Lecture Notes in Computer Science}, pages 387--404. Springer, 2020.

\bibitem{AkshayGKS17}
S.~Akshay, Paul Gastin, Shankara~Narayanan Krishna, and Ilias Sarkar.
\newblock Towards an efficient tree automata based technique for timed systems.
\newblock In {\em 28th International Conference on Concurrency Theory, {CONCUR}
  2017, September 5-8, 2017, Berlin, Germany}, pages 39:1--39:15, 2017.

\bibitem{alur1994theory}
Rajeev Alur and David~L Dill.
\newblock A theory of timed automata.
\newblock {\em Theoretical computer science}, 126(2):183--235, 1994.

\bibitem{BBLP04}
Gerd Behrmann, Patricia Bouyer, Kim~Guldstrand Larsen, and Radek Pel{\'{a}}nek.
\newblock Lower and upper bounds in zone based abstractions of timed automata.
\newblock In Kurt Jensen and Andreas Podelski, editors, {\em Tools and
  Algorithms for the Construction and Analysis of Systems, 10th International
  Conference, {TACAS} 2004, Held as Part of the Joint European Conferences on
  Theory and Practice of Software, {ETAPS} 2004, Barcelona, Spain, March 29 -
  April 2, 2004, Proceedings}, volume 2988 of {\em Lecture Notes in Computer
  Science}, pages 312--326. Springer, 2004.

\bibitem{bengtsson1995uppaal}
Johan Bengtsson, Kim Larsen, Fredrik Larsson, Paul Pettersson, and Wang Yi.
\newblock Uppaal—a tool suite for automatic verification of real-time
  systems.
\newblock In {\em International Hybrid Systems Workshop}, pages 232--243.
  Springer, 1995.

\bibitem{bouajjani1994automatic}
Ahmed Bouajjani, Rachid Echahed, and Riadh Robbana.
\newblock On the automatic verification of systems with continuous variables
  and unbounded discrete data structures.
\newblock In {\em International Hybrid Systems Workshop}, pages 64--85.
  Springer, 1994.

\bibitem{Bouyer04}
Patricia Bouyer.
\newblock Forward analysis of updatable timed automata.
\newblock {\em Formal Methods Syst. Des.}, 24(3):281--320, 2004.

\bibitem{BLR05}
Patricia Bouyer, Fran{\c{c}}ois Laroussinie, and Pierre{-}Alain Reynier.
\newblock Diagonal constraints in timed automata: Forward analysis of timed
  systems.
\newblock In Paul Pettersson and Wang Yi, editors, {\em Formal Modeling and
  Analysis of Timed Systems, Third International Conference, {FORMATS} 2005,
  Uppsala, Sweden, September 26-28, 2005, Proceedings}, volume 3829 of {\em
  Lecture Notes in Computer Science}, pages 112--126. Springer, 2005.

\bibitem{ClementeL15}
Lorenzo Clemente and Slawomir Lasota.
\newblock Timed pushdown automata revisited.
\newblock In {\em 30th Annual {ACM/IEEE} Symposium on Logic in Computer
  Science, {LICS} 2015, Kyoto, Japan, July 6-10, 2015}, page 738–749, 2015.

\bibitem{CL21}
Lorenzo Clemente and Slawomir Lasota.
\newblock Reachability relations of timed pushdown automata.
\newblock {\em J. Comput. Syst. Sci.}, 117:202--241, 2021.

\bibitem{CLLM17}
Lorenzo Clemente, Slawomir Lasota, Ranko Lazic, and Filip Mazowiecki.
\newblock Timed pushdown automata and branching vector addition systems.
\newblock In {\em 32nd Annual {ACM/IEEE} Symposium on Logic in Computer
  Science, {LICS} 2017, Reykjavik, Iceland, June 20-23, 2017}, pages 1--12.
  {IEEE} Computer Society, 2017.

\bibitem{Dang03}
Zhe Dang.
\newblock Pushdown timed automata: a binary reachability characterization and
  safety verification.
\newblock {\em Theor. Comput. Sci.}, (1-3):93–121, 2003.

\bibitem{GMS19}
Paul Gastin, Sayan Mukherjee, and B.~Srivathsan.
\newblock Fast algorithms for handling diagonal constraints in timed automata.
\newblock In {\em Computer Aided Verification - 31st International Conference,
  {CAV} 2019, New York City, NY, USA, July 15-18, 2019, Proceedings, Part {I}},
  volume 11561 of {\em Lecture Notes in Computer Science}, pages 41--59.
  Springer, 2019.

\bibitem{HKSW11}
Fr{\'{e}}d{\'{e}}ric Herbreteau, Dileep Kini, B.~Srivathsan, and Igor
  Walukiewicz.
\newblock Using non-convex approximations for efficient analysis of timed
  automata.
\newblock In {\em {IARCS} Annual Conference on Foundations of Software
  Technology and Theoretical Computer Science, {FSTTCS} 2011, December 12-14,
  2011, Mumbai, India}, volume~13 of {\em LIPIcs}, pages 78--89. Schloss
  Dagstuhl - Leibniz-Zentrum f{\"{u}}r Informatik, 2011.

\bibitem{tchecker}
Fr\'ed\'eric Herbreteau and Gerald Point.
\newblock Tchecker.
\newblock Available at \url{https://github.com/fredher/tchecker}.

\bibitem{HSTW20}
Fr{\'{e}}d{\'{e}}ric Herbreteau, B.~Srivathsan, Thanh{-}Tung Tran, and Igor
  Walukiewicz.
\newblock Why liveness for timed automata is hard, and what we can do about it.
\newblock {\em {ACM} Trans. Comput. Log.}, 21(3):17:1--17:28, 2020.

\bibitem{HSW12}
Fr{\'{e}}d{\'{e}}ric Herbreteau, B.~Srivathsan, and Igor Walukiewicz.
\newblock Better abstractions for timed automata.
\newblock In {\em Proceedings of the 27th Annual {IEEE} Symposium on Logic in
  Computer Science, {LICS} 2012, Dubrovnik, Croatia, June 25-28, 2012}, pages
  375--384. {IEEE} Computer Society, 2012.

\bibitem{Laarman13}
Alfons Laarman, Mads~Chr. Olesen, Andreas~Engelbredt Dalsgaard, Kim~Guldstrand
  Larsen, and Jaco van~de Pol.
\newblock Multi-core emptiness checking of timed {B}{\"{u}}chi automata using
  inclusion abstraction.
\newblock In Natasha Sharygina and Helmut Veith, editors, {\em Computer Aided
  Verification - 25th International Conference, {CAV} 2013, Saint Petersburg,
  Russia, July 13-19, 2013. Proceedings}, volume 8044 of {\em Lecture Notes in
  Computer Science}, pages 968--983. Springer, 2013.

\bibitem{larsen1997uppaal}
Kim~G Larsen, Paul Pettersson, and Wang Yi.
\newblock Uppaal in a nutshell.
\newblock {\em International journal on software tools for technology
  transfer}, 1(1-2):134--152, 1997.

\bibitem{Tripakis09}
Stavros Tripakis.
\newblock Checking timed b{\"{u}}chi automata emptiness on simulation graphs.
\newblock {\em {ACM} Trans. Comput. Log.}, 10(3):15:1--15:19, 2009.

\end{thebibliography}

\appendix
\section{Appendix}
\subsection{Benchmarks Used}\label{app:benchmarks}
We have used a total of 10 benchmarks. Of these the following are the ones which accept an empty language, $B_2(k)$ for all values of $k$, $B_4$, $B_6(k_1,k_2,k_3)$ for $k_1\geq k_2$, and $B_7$. The rest of the PDTA accept non-empty languages.
The benchmarks we have used are described as follows,
\begin{itemize}
    \item $B_1$: The benchmark $B_1$ in Figure~\ref{fig:b1} is an adaptation of the $PDTA$ in Figure~\ref{fig:tpda}. Instead of 3 pushes we have 8, and instead of, $y\leq3$ we have $y\leq10$. The state $q_1$ is reachable in the benchmark. We use this benchmark, since the zone graph for this $PDTA$ without any simulation will be infinitely large.
    \item $B_2(k)$: This is an adaptation of the $PDTA$ in Figure~\ref{fig:tpda2}, with $k+1$ pops between $q_0$, and $q_3$, and $y\leq k$ replacing $y\leq1$. The PDTA has $k+4$ states, and state $q_2$ is unreachable irrespective of the value of $k$. We have chosen this as a benchmark since in order to say that $q_2$ is unreachable, the loop $q_0\xrightarrow[]{}q_1\xrightarrow[]{}q_0$, has to be taken as many times as possible, creating more number of root nodes, with each root node containing the set of nodes that it can reach, making the size of the entire \verb|TLM| proportional to $k^2$. Therefore in order to accurately say that $q_2$ is unreachable, it would take $\mathcal{O}(k^2)$ time.
    \item $B_3(k_1,k_2)$: We have made this benchmark using a $TPDA$, with timed stack, and converting it to a timeless stack $PDTA$ using an approach similar to the one proposed in~\cite{ClementeL15}. The $TPDA$ with timed stack is shown in Figure~\ref{fig:timetountime1}. In both Figures~\ref{fig:timetountime1} and~\ref{fig:b3}, we can see that state $s_1$ is reachable only if $k_1\leq k_2$. This creates a significant difference in execution times when it comes to $B_3(4,3)$, and $B_3(3,4)$.
    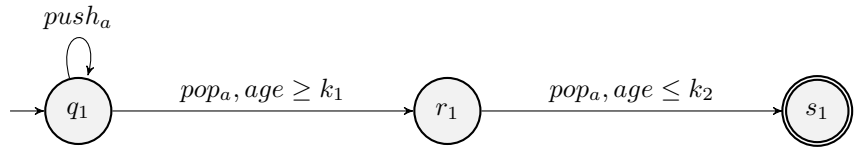
\begin{figure}[!h]
    \centering
        \begin{tikzpicture}[]
            \node[state, initial] (q) {$q_1$};
            \node[state, right=4cm of q] (r) {$r_1$};
            \node[state, accepting, right=4cm of r] (s) {$s_1$};
        
            \draw (q) edge[loop above] node[above]{$push_a$} (q)
            (q) edge[] node[above]{$pop_a,age\geq k_1$} (r)
            (r) edge[] node[above]{$pop_a,age\leq k_2$} (s);
        \end{tikzpicture}
        \caption{Timed Stack $TPDA$ used for conversion to $B_3(k_1,k_2)$.}
        \label{fig:timetountime1}
    \end{figure}
    \item $B_4$: This benchmark was originally used for testing in the region based implementation~\cite{AkshayGKS17}. Since there was only one push-pop pair, and no cycles with push/pop in it, it could be simulated using an extra clock $x_3$. The language accepted by the automaton is empty, since $q_5$ is unreachable.
    \item $B_5(k_1,k_2)$: An illustration of this benchmark is shown in Figure~\ref{fig:b5}, with $k_1=4$. $k_1$ should be even for $q_1$ to be reachable. This PDTA benchmark is used since a large number of pushes are required to reach $q_1$, creating more root nodes. Also because of internal transition loops between $r_x$, and $r_x'$ in the PDTA which can create many non equivalent zones, this creates the possibility of a large \verb|TLM|.
    \item $B_6(k_1,k_2,k_3)$: The illustration for this benchmark is shown in Figure~\ref{fig:b6}. It has been used to indicate that the size of the TLM is not just dependent on the size of the automaton, but also on the constants used in the automaton. In this PDTA, the size of the \verb|TLM| is dependent on all $k_1$, $k_2$, and $k_3$. $k_1$, and $k_2$ enforce the pushes and pops in the automaton to be taken a fixed number of times, and if $k_1\geq k_2$,  the language accepted is empty.
    \item $B_7$: The benchmark illustrated in Figure~\ref{fig:b7} has been used also in order to highlight the significance of open guards. The guards over the pushes, $push_a$, and $push_b$, make it impossible to take two pushes on $a$, then a push on $b$, which is required to empty the stack on reaching $q_2$. And the guards $x==0\land z==20$, also make it harder to directly reach $q_2$ with an empty stack.
    \item $B_8$: A variation of this $PDTA$ was used in testing for the region based implementation as well. States $\{q_1,q_3,q_5,q_6,q_8\}$ are reachable, and the language is not empty.
    \item $B_9(k_1,k_2)$: The PDTA for this benchmark with $k_1=2$ is illustrated in Figure~\ref{fig:b9}. This benchmark has been used in order to force the automaton to take only a single path which is allowed. Also the extra 2 state loops on nodes like $r_2$, can create many non-equivalent nodes under one root node, and these nodes can then be propagated to other root nodes via pushes and pops, which can make the execution time heavily dependent on the constant $k_2$.
    \item $B_{10}$: This benchmark has been used in order to show that our tool can also handle open constraints on guards. This benchmark like $B_6$ one can have the size of \verb|TLM| highly dependent on the constants in transition guards.
\end{itemize}
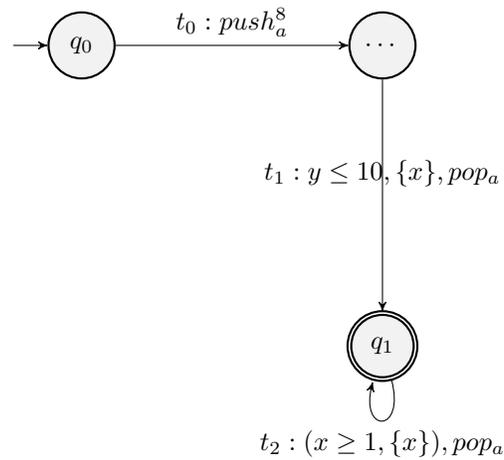
\begin{figure}[H]
\centering
\begin{tikzpicture}[node distance = 4cm]
    \node[state, initial] (q0) {$q_0$};
    \node[state, right of=q0] (cdots) {$\cdots$};
    \node[state, accepting, below of=cdots] (q1) {$q_1$};

    \draw (q0) edge[] node[above]{$t_0:push_a^{8}$} (cdots)
    (cdots) edge[] node[above]{$t_1:y\leq10,\{x\},pop_a$} (q1)
    (q1) edge[loop below] node[below]{$t_2:(x\geq 1,\{x\}),pop_a$} (q1);
\end{tikzpicture}
\caption{$B_1 : PDTA$. States $q_0$, and $q_1$ are reachable. The $\cdots$ indicates that there are a series of states, say, $\{r_1,r_2,...r_8\}$ between $q_0$, and $q_1$, with 8 pushes, starting from $q_0\xrightarrow[]{}r_1$, to $r_7\xrightarrow[]{}r_8$. And in the end there transition $t_1$ is in between $r_8$ and $q_1$.}
\label{fig:b1}
\end{figure}
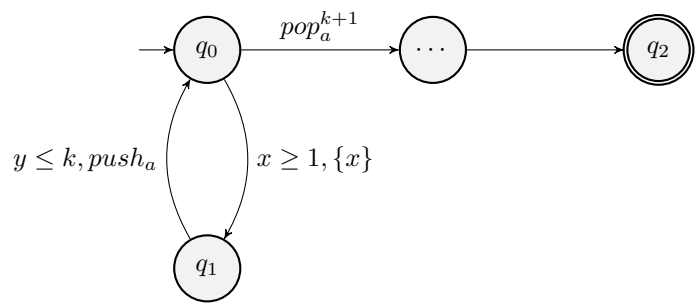
\begin{figure}[H]
\centering
\begin{tikzpicture}[node distance = 3cm]
    \node[state, initial] (q0) {$q_0$};
    \node[state, right of=q0] (cdots) {$\cdots$};
    \node[state, accepting, right of=cdots] (q2) {$q_2$};
    \node[state, below=2cm of q0] (q1) {$q_1$};

    \draw (q0) edge[bend left] node[right]{$x\geq 1,\{x\}$} (q1)
    (q1) edge[bend left] node[left]{$y\leq k,push_a$} (q0)
    (q0) edge[] node[above]{$pop_a^{k + 1}$} (cdots)
    (cdots) edge[] node[above]{} (q2);
\end{tikzpicture}
\caption{Parametrized $PDTA$ - $B_2(k)$. Used in experiments using different values of $k$. The $\cdots$ indicate there are a series of transitions (linear) between $q_0$ and $q_2$ consisting of $k+1$ pops. The number of states between $q_0$ and $q_2$, will be $k+1$.}
\label{fig:b2}
\end{figure}
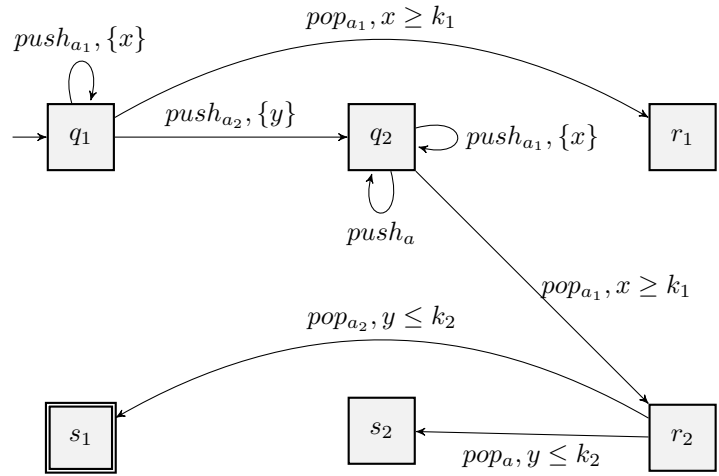
\begin{figure}[H]
\centering
\begin{tikzpicture}[node distance = 4cm]
    \node[state, rectangle, initial] (q1) {$q_1$};
    \node[state, rectangle, right of=q1] (q2) {$q_2$};
    \node[state, rectangle, right of=q2] (r1) {$r_1$};
    \node[state, rectangle, below of=r1] (r2) {$r_2$};
    \node[state, accepting, rectangle, below of=q1] (s1) {$s_1$};
    \node[state, rectangle, below=3cm of q2] (s2) {$s_2$};
    
    \draw (q1) edge[] node[above]{$push_{a_2},\{y\}$} (q2)
    (q1) edge[loop above] node[above]{$push_{a_1}, \{x\}$} (q1)
    (q2) edge[loop below] node[below]{$push_a$} (q2)
    (q2) edge[loop right] node[right]{$push_{a_1},\{x\}$} (q2)
    (q1) edge[bend left] node[above]{$pop_{a_1},x\geq k_1$} (r1)
    (q2) edge[] node[right]{$pop_{a_1},x\geq k_1$} (r2)
    (r2) edge[] node[below]{$pop_a,y \leq k_2$} (s2)
    (r2) edge[bend right] node[above]{$pop_{a_2}, y \leq k_2$} (s1);
\end{tikzpicture}
\caption{$B_3(k_1, k_2)$ Parametrized $PDTA$ with parameters $k_1$, and $k_2$. If $k_1 > k_2$, then the state $\{q_1,r_1\}$ is reachable. Otherwise, $\{q_1,r_1,s_1\}$ are reachable.}
\label{fig:b3}
\end{figure}
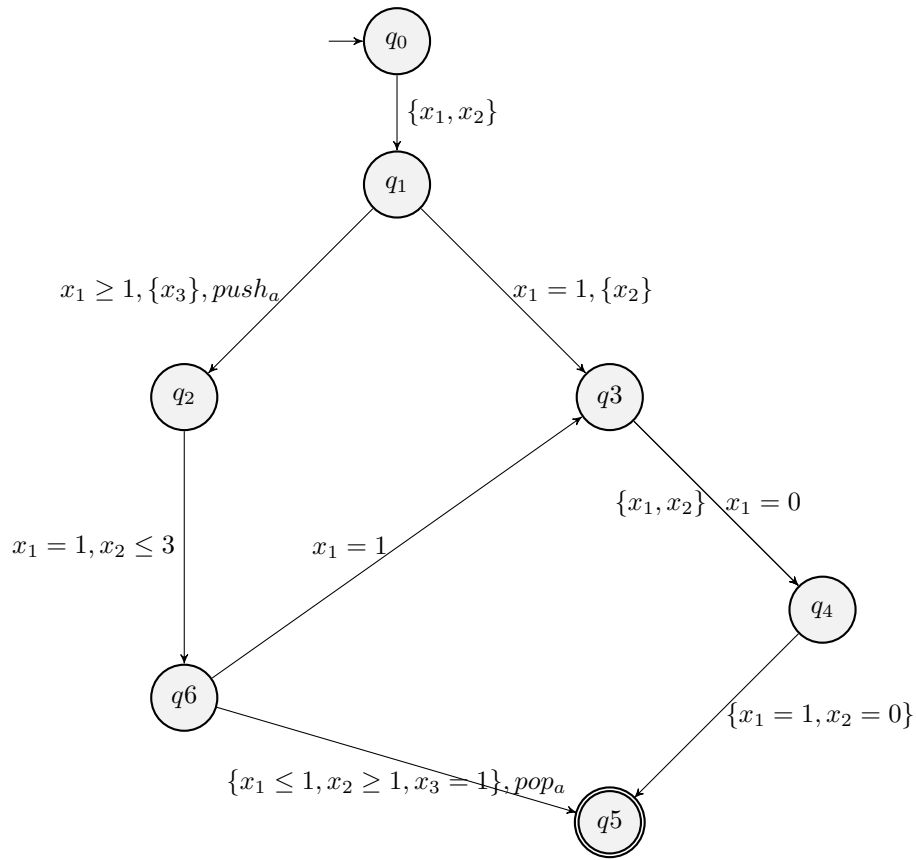
\begin{figure}[H]
\centering
\begin{tikzpicture}[node distance = 4cm]
    \node[state, initial] (q0) {$q_0$};
    \node[state, below=1cm of q0] (q1) {$q_1$};
    \node[state, below left of=q1] (q2) {$q_2$};
    \node[state, below right of=q1] (q3) {$q3$};
    \node[state, below right of=q3] (q4) {$q_4$};
    \node[state, accepting, below left of=q4] (q5) {$q5$};
    \node[state, below of=q2] (q6) {$q6$};
    
    \draw (q0) edge[] node[right]{$\{x_1,x_2\}$} (q1)
    (q1) edge[] node[left]{$x_1\geq 1, \{x_3\},push_a$} (q2)
    (q1) edge[] node[right]{$x_1=1,\{x_2\}$} (q3)
    (q2) edge[] node[left]{$x_1=1,x_2\leq3$} (q6)
    (q6) edge[] node[left]{$x_1=1$} (q3)
    (q6) edge[] node[below]{$\{x_1\leq1,x_2\geq 1,x_3=1\},pop_a$} (q5)
    (q3) edge[] node[left]{$\{x_1,x_2\}$} (q4)
    (q4) edge[] node[right]{$\{x_1=1,x_2=0\}$} (q5)
    (q3) edge[] node[right]{$x_1=0$} (q4);
\end{tikzpicture}
\caption{$B_4$ $PDTA$, states $\{q_0,q_1,q_3,q_4\}$ are reachable.}
\label{fig:b4}
\end{figure}
\begin{figure}[H]
\centering
\begin{tikzpicture}[node distance = 3cm]
    \node[state, initial] (q0) {$q_0$};
    \node[state, right=1cm of q0] (r1) {$r_1$};
    \node[state, below of=r1] (r1p) {$r_1'$};
    \node[state, right=2cm of r1] (r2) {$r_2$};
    \node[state, above of=r2] (r2p) {$r_2'$};
    \node[state, right=2cm of r2] (r3) {$r_3$};
    \node[state, below of=r3] (r3p) {$r_3'$};
    \node[state, right=2cm of r3] (r4) {$r_4$};
    \node[state, above of=r4] (r4p) {$r_4'$};
    \node[state, accepting, right=1cm of r4] (q1) {$q_1$};
    
    \draw (q0) edge[] node[above]{$push_a$} (r1)
    (r1) edge[bend left] node[right]{$x\geq1,\{x\}$} (r1p)
    (r1p) edge[bend left] node[left]{$y\leq k_2$} (r1)
    (r1) edge[] node[above]{$push_a$} (r2)
    (r2) edge[bend left] node[left]{$x\geq1,\{x\}$} (r2p)
    (r2p) edge[bend left] node[right]{$y\leq k_2$} (r2)
    (r2) edge[] node[above]{$pop_a$} (r3)
    (r3) edge[bend left] node[right]{$x\geq1,\{x\}$} (r3p)
    (r3p) edge[bend left] node[left]{$y\leq k_2$} (r3)
    (r3) edge[] node[above]{$pop_a$} (r4)
    (r4) edge[bend left] node[left]{$x\geq1,\{x\}$} (r4p)
    (r4p) edge[bend left] node[right]{$y\leq k_2$} (r4)
    (r4) edge[] node[above]{} (q1);
\end{tikzpicture}
\caption{$B_5(4,k_2)$: Parametrized PDTA with $k_1=4$. $k_1$ indicates the total number of pairs of states $r_x$, and $r_x'$ in between $q_0$ and $q_1$. For all even $k_1$, state $q_1$ is reachable.}
\label{fig:b5}
\end{figure}
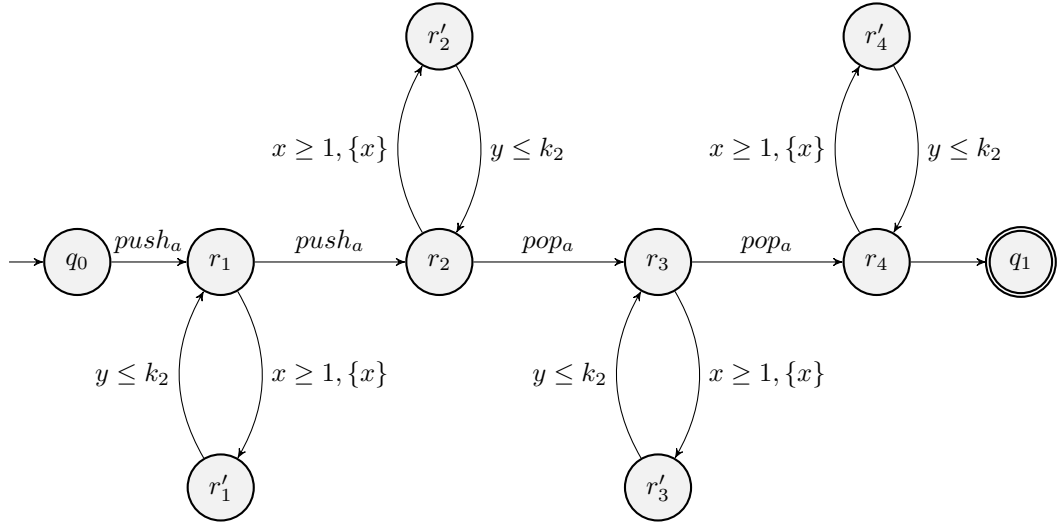
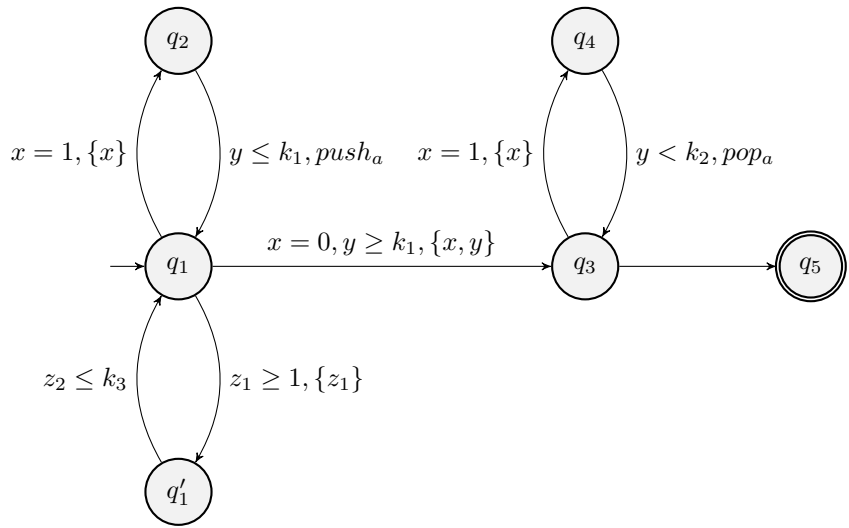
\begin{figure}[H]
\centering
\begin{tikzpicture}[node distance = 3cm]
    \node[state, initial] (q1) {$q_1$};
    \node[state, below of=q1] (q1p) {$q_1'$};
    \node[state, above of=q1] (q2) {$q_2$};
    \node[state, right=4.5cm of q1] (q3) {$q_3$};
    \node[state, above of=q3] (q4) {$q_4$};
    \node[state, accepting, right of=q3] (q5) {$q_5$};
    
    \draw (q1) edge[bend left] node[left]{$x=1,\{x\}$} (q2)
    (q2) edge[bend left] node[right]{$y\leq k_1,push_a$} (q1)
    (q1) edge[bend left] node[right]{$z_1\geq 1, \{z_1\}$} (q1p)
    (q1p) edge[bend left] node[left]{$z_2\leq k_3$} (q1)
    (q1) edge[] node[above]{$x=0,y\geq k_1,\{x,y\}$} (q3)
    (q3) edge[bend left] node[left]{$x=1,\{x\}$} (q4)
    (q4) edge[bend left] node[right]{$y<k_2,pop_a$} (q3)
    (q3) edge[] node[]{} (q5);
\end{tikzpicture}
\caption{$B_6(k_1,k_2,k_3)$: State $q_5$ is not reachable if $k_1 \geq k_2$, otherwise it is.}
\label{fig:b6}
\end{figure}
\begin{figure}[H]
\centering
\begin{tikzpicture}[node distance = 3cm]
    \node[state, initial] (q1) {$q_1$};
    \node[state, right=4cm of q1] (q2) {$q_2$};
    \node[state, above left of=q2] (q3) {$q_3$};
    \node[state, above right of=q2] (q4) {$q_4$};
    \node[state, accepting, right=4cm of q2] (q5) {$q_5$};
    
    \draw (q1) edge[loop above] node[above]{$x>1,push_a,\{x\}$} (q1)
    (q1) edge[loop below] node[below]{$y<2,push_b,\{y\}$} (q1)
    (q1) edge[] node[above]{$x=0,z=20$} (q2)
    (q2) edge[] node[left]{$pop_b$} (q3)
    (q3) edge[] node[above]{$pop_a$} (q4)
    (q4) edge[] node[right]{$pop_a$} (q2)
    (q2) edge[] node[]{} (q5);
\end{tikzpicture}
\caption{$B_7$: $PDTA$ with state $\{q_1\}$ as reachable.}
\label{fig:b7}
\end{figure}
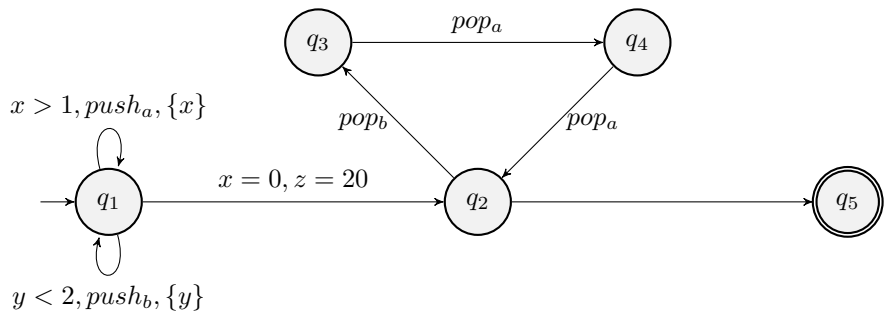
\begin{figure}[H]
\centering
\begin{tikzpicture}[node distance = 3cm]
    \node[state, initial] (q1) {$q_1$};
    \node[state, right of=q1] (q2) {$q_2$};
    \node[state, right=3.5cm of q2] (q3) {$q_3$};
    \node[state, below=1cm of q3] (q4) {$q_4$};
    \node[state, left=3.5cm of q4] (q5) {$q_5$};
    \node[state, left of=q5] (q6) {$q_6$};
    \node[state, below=1cm of q6] (q7) {$q_7$};
    \node[state, accepting, right of=q7] (q8) {$q_8$};
    
    \draw (q1) edge[] node[above]{$push_a,\{x_2\}$} (q2)
    (q2) edge[] node[above]{$x_2=1,pop_a,\{x_4\}$} (q3)
    (q3) edge[] node[right]{$x_4=0,push_b,\{x_3\}$} (q4)
    (q4) edge[] node[above]{$x_3\geq1,pop_b,\{x_1\}$} (q5)
    (q5) edge[] node[above]{$\{x_1\}$} (q6)
    (q6) edge[] node[left]{$push_a,\{x_2\}$} (q7)
    (q7) edge[] node[above]{$x_2\geq1,pop_a$} (q8);
\end{tikzpicture}
\caption{$B_8$: $PDTA$ with states $\{q_1,q_3,q_5,q_6,q_8\}$ as reachable.}
\label{fig:b8}
\end{figure}
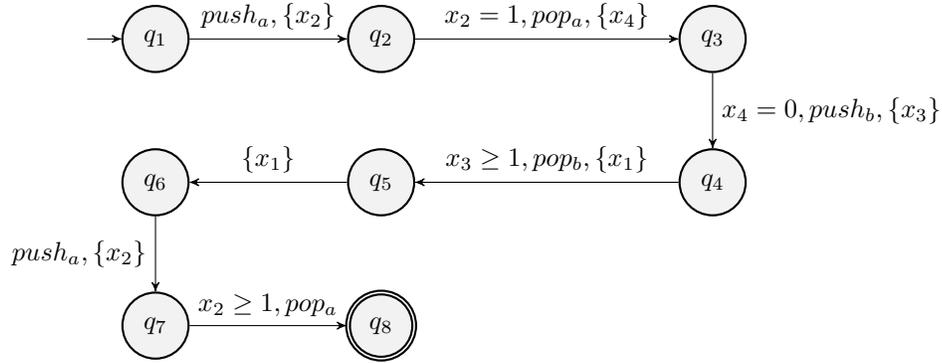
\begin{figure}[H]
\centering
\begin{tikzpicture}[node distance = 3cm]
    \node[state, initial] (q0) {$q_0$};
    \node[state, above left=2cm of q0] (r1) {$r_1$};
    \node[state, above right=2cm of r1] (r2) {$r_2$};
    \node[state, above=1cm of r2] (r2p) {$r_2'$};
    \node[state, below right=2cm of r2] (r3) {$r_3$};
    \node[state, below left=2cm of q0] (r4) {$r_4$};
    \node[state, below right=2cm of r4] (r5) {$r_5$};
    \node[state, below=1cm of r5] (r5p) {$r_5'$};
    \node[state, above right=2cm of r5] (r6) {$r_6$};
    \node[state, right of=q0] (s1) {$s_1$};
    \node[state, right of=s1] (s2) {$s_2$};
    \node[state, right of=s2] (cdots) {$\cdots$};
    \node[state, accepting, right of=cdots] (sf) {$s_f$};
    
    \draw (q0) edge[] node[below left]{$push_{a_1}$} (r1)
    (r1) edge[] node[above left]{$push_{a_2}$} (r2)
    (r2) edge[] node[right]{$push_{a_3}$} (r3)
    (r3) edge[] node[left]{$push_{a_4}$} (q0)
    (q0) edge[] node[left]{$push_{a_5}$} (r4)
    (r4) edge[] node[left]{$push_{a_6}$} (r5)
    (r5) edge[] node[right]{$push_{a_7}$} (r6)
    (r6) edge[] node[right]{$push_{a_8}$} (q0)
    (q0) edge[] node[above]{$pop_{a_4}$} (s1)
    (s1) edge[] node[above]{$pop_{a_3}$} (s2)
    (s2) edge[] node[above]{$pop_{a_2}$} (cdots)
    (cdots) edge[] node[above]{$pop_{a_5}$} (sf)
    (r2) edge[bend left] node[left]{$x\geq1,\{x\}$} (r2p)
    (r2p) edge[bend left] node[right]{$y\leq k_2$} (r2)
    (r5) edge[bend left] node[right]{$x\geq1,\{x\}$} (r5p)
    (r5p) edge[bend left] node[left]{$y\leq k_2$} (r5);
\end{tikzpicture}
\caption{$B_9(2,k_2)$: Parametrized PDTA with $k_1=2$. $k_1$ indicates the total number of loops of size 4, around $q_0$. It involves a direct loop of 4 pushes on unique symbols $a_i$ to $a_{i+4}$, and another loop on the middle state. Finally starting from $q_0$ there is a line of transitions with no loops, having pops matching all loops in order. Only $\{q_0,s_f\}$ are reachable in the automata irrespective of the values of $k_1,k_2$.}
\label{fig:b9}
\end{figure}
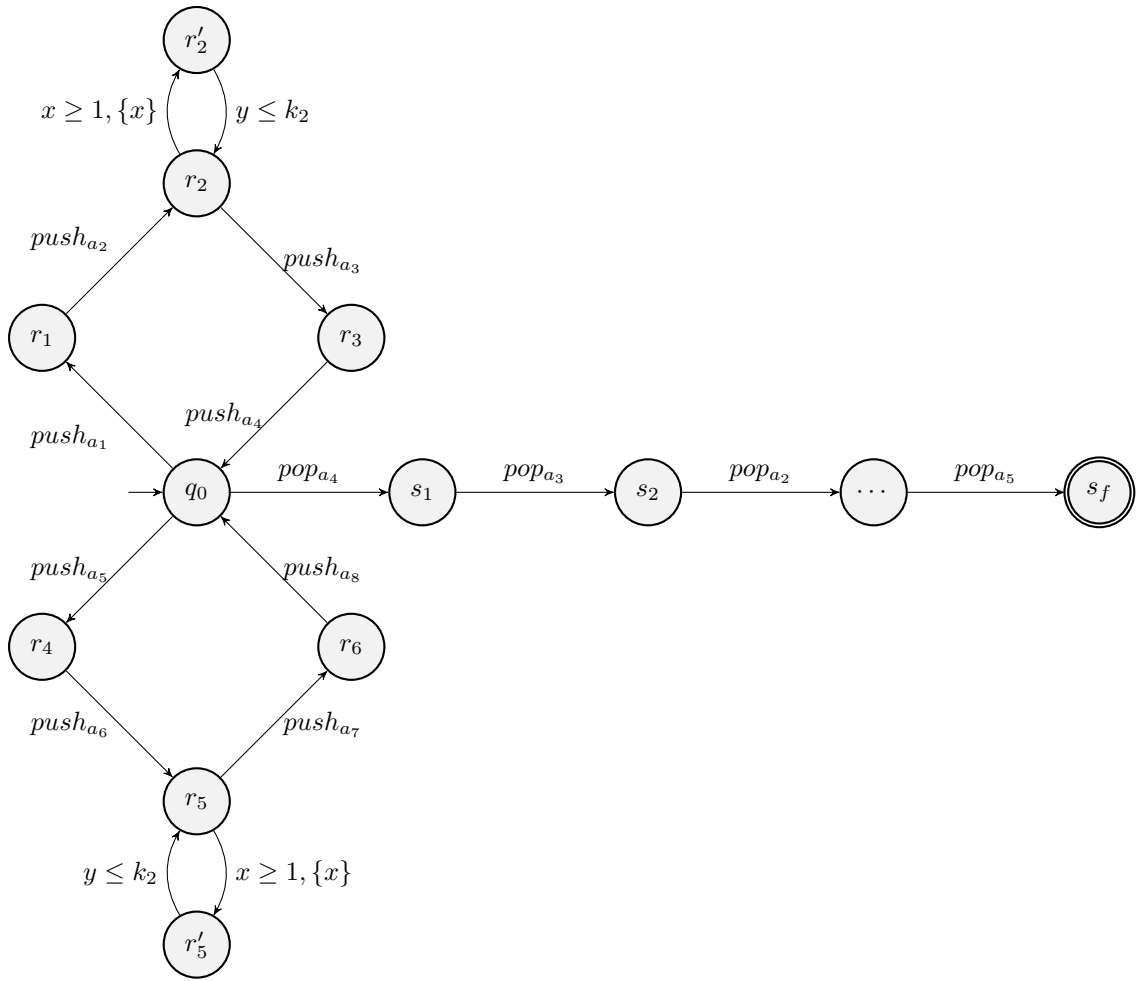
\begin{figure}[H]
\centering
\begin{tikzpicture}[node distance = 3cm]
    \node[state, initial] (q1) {$q_1$};
    \node[state, right of=q1] (q2) {$q_2$};
    \node[state, above of=q2] (q3) {$q_3$};
    \node[state, accepting, right of=q2] (q4) {$q_4$};
    
    \draw (q1) edge[loop above] node[above]{$x>1,push_a,\{x\}$} (q1)
    (q1) edge[loop below] node[below]{$y<2,push_b,\{y\}$} (q1)
    (q1) edge[] node[above]{$x=0,z=4$} (q2)
    (q2) edge[bend left] node[left]{$pop_a$} (q3)
    (q3) edge[bend left] node[right]{$pop_b$} (q2)
    (q2) edge[] node[]{} (q4);
\end{tikzpicture}
\caption{$B_{10}$: $PDTA$ with states $\{q_1,q_2,q_3,q_4\}$ as reachable.}
\label{fig:b10}
\end{figure}
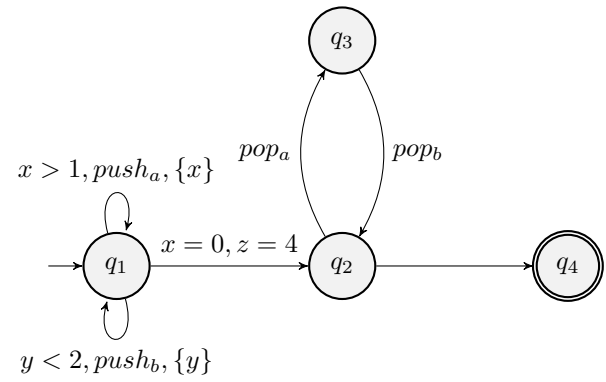
\newpage
\subsection{All Results}\label{app:all_results}
Here we display all results on the benchmarks that have been used.

\begin{table}[H]
    \centering
      \begin{tabular}{ |p{2.3cm}||p{1.5cm}|p{1.5cm}|p{1.5cm}|p{1.5cm}|p{1.5cm}|p{1.5cm}|  }
 \hline
 \multicolumn{7}{|c|}{Testing on Benchmarks} \\
 \hline
  Benchmark& \quad $\preceq_{LU}$ & \quad $\preceq_{LU}$ & \quad $\sim_{LU}$ &\quad $\sim_{LU}$& \quad $RB$ & \quad $RB$\\
 & Time& \# nodes & Time & \# nodes& Time & \# nodes\\
 \hline
 $B_1$&   0.2  & 17& 0.2 &17 &235.6 &4100\\
 $B_2(5)$   & 0.3    &27&   0.3 &27 & 21 &1500\\
 $B_2(10)$   & 0.8    &77&   0.8 &77 & 6835.8 &30200\\
 $B_2(100)$&   20.0  & 5252   &20.7 &5252 & T.O. & $\geq$154700\\
 $B_2(1000)$ & 9140.4 & 502502&  9164.8 &502502 & T.O. & T.O.\footnotemark[1]\\
 $B_3(4,3)$& 0.2  & 6   & 0.2 & 6 & 1043.8 & 14300\\
 $B_3(3,4)$& 0.2  & 9 &0.2 &9 & 98.8 & 3400\\
 $B_4$    &0.2 & 8&  0.1 & 8&0.3 &17 \\
 $B_5(100,10)$& 0.8 & 202 & 5.4 &2212 &OOM & OOM\\
 $B_5(100,100)$&0.6 & 202 & 67.6 &20302 &OOM & OOM\\
 $B_5(100,1000)$& 0.7 & 202 & 3564.3 &201202 &OOM & OOM\\
 $B_5(1000,100)$& 4.2& 2002 & 673.8 &202102 &OOM &OOM\\
 $B_5(5000,100)$& 23.2& 10002 &3429.3 &1010102 &OOM &OOM\\
 $B_6(4,5,100)$& 0.3 & 30&12.4 &2459 &NA &NA\\
 $B_6(4,5,1000)$& 0.3 & 30 &483.1 &24059 &NA &NA\\
 $B_6(4,5,10000)$& 0.3 & 30 &47694.8 &240059 &NA &NA\\
 $B_6(5,4,100)$& 0.3 & 30 &14.3 &3047 &NA &NA\\
 $B_6(5,4,1000)$& 0.3 & 30 &611.8 &30047 &NA &NA\\
 $B_6(5,4,10000)$& 0.3 & 30 &60271.9 &300047 &NA &NA\\
 $B_6(500,501,100)$& 38.9 & 3006 &509.8 &34802 &NA &NA\\
 $B_6(501,500,100)$& 38.2 & 3006 &501.0 &34799 &NA &NA\\
 $B_7$& 112.4 & 4475 & 113.1 &4475 &NA &NA\\
 $B_8$& 0.2  & 8 & 0.2 & 8 & 0.6 &26\\
 $B_9(10,10)$& 0.3& 81 &13.3 &4136 &T.O. &$\geq$96500\\
 $B_9(50,10)$& 1.1& 401&52.9 &20856 &OOM &OOM\\
 $B_9(100,10)$& 1.9 & 801 &107.0 &41756 &OOM &OOM\\
 $B_9(10,20)$& 0.5 & 81 &47.2 &14091 &T.O. & $\geq$106400\\
 $B_9(10,50)$& 0.4 & 81 &543.2 &79356 &T.O. &$\geq$370500\\
 $B_9(10,100)$& 0.4 & 81 &5374.1 &306131 &T.O. &$\geq$133600\\
 $B_{10}$& 1.8 & 150 &2.0 &166 &NA &NA\\
 \hline
\end{tabular}
    \caption{List of all results on the three environments. All time is recorded in milliseconds. T.O. refers to timeout before 120 seconds, and OOM refers to OOM kill of process. In case of timeout $\geq n$, refers to recorded number of nodes before timeout occurred. In case of $B_2(1000)$ the preprocessing is not complete before the timeout and hence no nodes are displayed.}
    \label{tab:final_expanded_results}
\end{table}
\end{document}